\documentclass[letterpaper, 10 pt, conference]{ieeeconf}  

\IEEEoverridecommandlockouts                              
\usepackage{graphics} 
\usepackage{mathptmx} 
\usepackage{amsmath} 
\usepackage{amssymb}  
\usepackage{float}
\usepackage{graphicx}
\usepackage[hidelinks]{hyperref}
\newcommand{\R}{\mathcal{R}}
\let\labelindent\relax
\usepackage{enumitem}
\usepackage{subcaption}
\usepackage{hyperref}
\newtheorem{theorem}{Theorem}[section]
\newtheorem{definition}{Definition}[section]

\newtheorem{lemma}[theorem]{Lemma}

\title{\LARGE \bf
Rules of the Road: Safety and Liveness Guarantees for Autonomous Vehicles
}

\author{Karena X. Cai, Tung Phan-Minh, Soon-Jo Chung, Richard M. Murray
\thanks{This research supported by the National Science Foundation award CNS-1545126.}
}

\begin{document}

\maketitle
\thispagestyle{empty}
\pagestyle{empty}

\begin{abstract}
The ability to guarantee safety and progress for all vehicles is vital to the success of the autonomous vehicle industry. We present a framework for designing autonomous vehicle behavior in a way that is safe and guarantees progress for all agents. In this paper, we first introduce a new game paradigm which we term the quasi-simultaneous game. We then define an agent protocol that all agents must use to make decisions in this quasi-simultaneous game setting. According to the protocol, agents first select an intended action using a behavioral profile. Then, the protocol defines whether an agent has precedence to take its intended action or must take a sub-optimal action. The protocol ensures safety under all traffic conditions and liveness for all agents under `sparse' traffic conditions. We provide proofs of correctness of the protocol and validate our results in simulation.
\end{abstract}

\section{Introduction}
A prerequisite for introducing autonomous vehicles into our society is a compelling proof of their safety and efficacy. Unfortunately, designing agent strategies in interactive multi-agent settings is extremely difficult since agent behavior is highly coupled and the computational complexity grows exponentially when reasoning about joint action spaces.  

Most approaches for designing agent behavior focus on designing an individual agent's strategy while modeling interactions with other agents using some interactive behavioral model. Minimum violation motion-planning has been proposed to help the vehicle choose the trajectory that minimizes violation of a set of ordered rules \cite{tumova2013least, IncSynth2013}. Rulebooks are a way to set priorities among possibly conflicting sets of specifications \cite{2019Rulebooks}. The game-theoretic approach has been to model agent decision-making as interacting partially-observable Markov Decision Processes (POMDPs) \cite{boutilier1999sequential, gmytrasiewicz2005framework}. These methods often capture the reactivity of agents by modeling a reward function defined on a joint action space but suffers from the curse of dimensionality. Data-driven methods are used to learn interactive models between agents and design an optimal strategy for an individual agent based on this learned model \cite{sadigh2016planning, sadigh2017active}. When designing an \textit{individual} agent strategy, how other agents are assumed to be behaving is not explicitly defined—thereby preventing the ability to make complete safety guarantees. 

Instead of reasoning about safety on the individual agent level, the authors in \cite{shoham1995social} introduce the idea of reasoning about safety as a property of the collective of agents. In particular, they introduce the idea of social laws, which are a set of rules imposed upon all agents in a multiagent system to ensure some desirable global behaviors like safety or progress \cite{shoham1995social,van2007social}. The design of social laws is intended to achieve the desirable global behavioral properties in a minimally-restrictive way \cite{shoham1995social}. The problem of automatically synthesizing useful social laws for a set of agents for a general state space, however, has been shown to be NP-complete~\cite{shoham1995social}. Model checking tools have also been designed to verify correctness of agent protocols for multi-agent systems, but these do not solve the protocol synthesis problem \cite{lomuscio2017mcmas, van2007social}. The Responsibility-Sensitive-Safety (RSS) framework \cite{2017FormalModel} adopts a similar top-down philosophy for guaranteeing safety by providing a set of rules like maintaining distance, yielding, etc, but does not provide guarantees of agent progress.  

Similarly, the Assume-Guarantee framework for autonomous vehicles introduced in \cite{agprofiles} dictates all agents must abide by some behavioral contract where agents make decisions according to a behavioral profile. With all agents operating according to the behavioral profile, the interactions are not necessarily coordinated. In particular, there might be multiple agents with conflicting goals. The process for resolving multiple conflicting processes in a local, decentralized manner is addressed in the Drinking Philosopher problem, which provides a mechanism for resolving conflicts by defining a local, decentralized algorithm for assigning precedence among agents \cite{chandy1984drinking}. 
We introduce an agent protocol that is an adaptation of the Drinking Philosopher problem. The agent protocol is defined so agents use a behavioral profile to select an intended action. Additional constraints specified in the profile, determine when an agent has precedence in taking its intended action. Unlike \cite{sahin2020drinking}, our framework leverages the structure of the driving road network and takes into account the inertial properties of agents. 

The main contributions of this paper are as follows:
1) The introduction of a new game paradigm, which we term the quasi-simultaneous discrete-time multi-agent game, 2) the definition of an agent protocol that defines local rules agents must use to select their actions, 3) safety and liveness proofs when all agents operate according to these local rules and 4) simulations as proof of concept of the safety and liveness guarantees.

\section{Quasi-Simultaneous Discrete-Time Game}
\label{section_game}
We propose a quasi-simultaneous discrete-time game paradigm, which we motivate by looking at the shortcomings of more traditional game paradigms. In synchronous games, all agents in the game are making decisions simultaneously. Since agents are making decisions in the absence of other agent behaviors, it does not capture the sequential nature of real-life decision making. Turn-based games offer potential for capturing sequential decision-making, but the turns are often assigned arbitrarily. The quasi-simultaneous discrete-time game offers a way to assign turns, but in a turn order based on the agent states defined with respect to the road network.

A \textit{state} associated with a set of variables is an assignment of values to those variables. A game evolves by a sequence of state changes. A quasi-simultaneous game has the following two properties regarding state changes: 1) each agent will get to take a turn in each time-step of the game and 2) each agent must make their turn in an order that emerges from a locally-defined precedence assignment algorithm. We define a quasi-simultaneous game where all agents act in a local, decentralized manner as follows $\mathfrak{G} = \langle \mathfrak{A}, \mathcal{Y}, Act_{[\cdot]}, P \rangle$, where $\mathfrak{A}$ is the set of all agents in the game, $\mathcal{Y}$ is the set of all variables in the game, $Act_{\text{Ag}}$ be the set of all possible actions $\text{Ag}$ can take. Finally, $P: \mathcal{Y} \rightarrow \text{PolyForest}(\mathfrak{A})$, is the precedence assignment function where $\text{PolyForest}$ is an operator that maps a set to a polyforest graph object. The polyforest, with its nodes and directed edges, defines the global turn order (of precedence) of the set of all agents based on the agent states. 


\section{Specific Agent Class}
In order to make global guarantees on safety and progress, we first only consider a single specific class of agents whose attributes, dynamics, motion-planner, and perception capabilities are described in more detail in the following section. Although assuming a single class of agents seems very restrictive, the work can be easily extended to accommodate additional variants of the agent class. These extensions, however, are beyond the scope of this work. 

\subsection{Agent Attributes}
\label{agent_attributes}
Each \textit{agent} $\text{Ag}$ is characterized by a set of variables $\mathcal{V}_{\text{Ag}} \subseteq \mathcal{Y}$. We define  $\{\texttt{Id}_{\text{Ag}}, \texttt{Tc}_{\text{Ag}}, \texttt{Goal}_{\text{Ag}} \}~\subseteq~{\mathcal{V}}_{\text{Ag}}$ where $\texttt{Id}_{\text{Ag}}$, $\texttt{Tc}_{\text{Ag}}$, and $\texttt{Goal}_{\text{Ag}}$ are the agent's ID number, token count and goal respectively. The token count and ID are defined in greater depth in Section \ref{section_conflict_cluster_resolution}. Agents are assumed to have the capability of querying the token counts of neighboring agents. 

In this paper, we only consider \textit{car} agents such that if $\text{Ag} \in \mathfrak{A}$, then $\mathcal{V}_{\text{Ag}}$ includes $x_{\text{Ag}}$, $y_{\text{Ag}}$, $\theta_{\text{Ag}}$, $v_{\text{Ag}}$, namely its absolute coordinates, heading and velocity. We let $S_{\text{Ag}}$ denote the set that contains all possible states of these variables in $\mathcal{V}_{\text{Ag}}$. $\mathcal{V}_{\text{Ag}}$ also has parameters:
$ {a_{\text{min}}}_{\text{Ag}} \in \mathbb{Z}, {a_{\text{max}}}_{\text{Ag}}\in \mathbb{Z}, {v_{\text{min}}}_{\text{Ag}}\in \mathbb{Z} \text{ and } {v_{\text{max}}}_{\text{Ag}}\in \mathbb{Z}$ which define the minimum and maximum accelerations and velocities respectively. The agent control actions are defined by two parameters: 1) an acceleration value $\text{acc}_{\text{Ag}}$ between ${a_{\text{min}}}_{\text{Ag}}$ and ${a_{\text{max}}}_{\text{Ag}}$ and 2) a steer maneuver $\gamma_{\text{Ag}} \in $\{\texttt{left-turn}, \texttt{right-turn}, \texttt{left-lane change}, \texttt{right-lane change}, \texttt{straight}\}. 

The discrete agent dynamics works as follows. At a given state $s \in S_{\text{Ag}}$ at time $t$, for a given control action $(\text{acc}_{Ag}, \gamma_{Ag})$, the agent first applies the acceleration to update its velocity $s.v_{Ag,t+1} = s.v_{Ag,t} + \text{acc}_{\text{Ag}}$. Once the velocity is applied, the steer maneuver (if at the proper velocity) is taken and the agent occupies a set of grid-points, specified in Fig. \ref{fig_agent_dynamics}, while taking its maneuver. The agent state-transition function $\tau_{\text{Ag}}: S_{\text{Ag}} \times Act_{\text{Ag}} \rightarrow S_{\text{Ag}}$ defines the state an agent will transition to by taking an action $a$ at a given state $s_{\text{Ag}}$ and the state precondition $\rho_{\text{Ag}}:S_{\text{Ag}} \rightarrow 2^{Act_{\text{Ag}}}$ functions defines the set of allowable actions at a given state. 
\begin{figure}[H]
\centering
\includegraphics[scale=0.23]{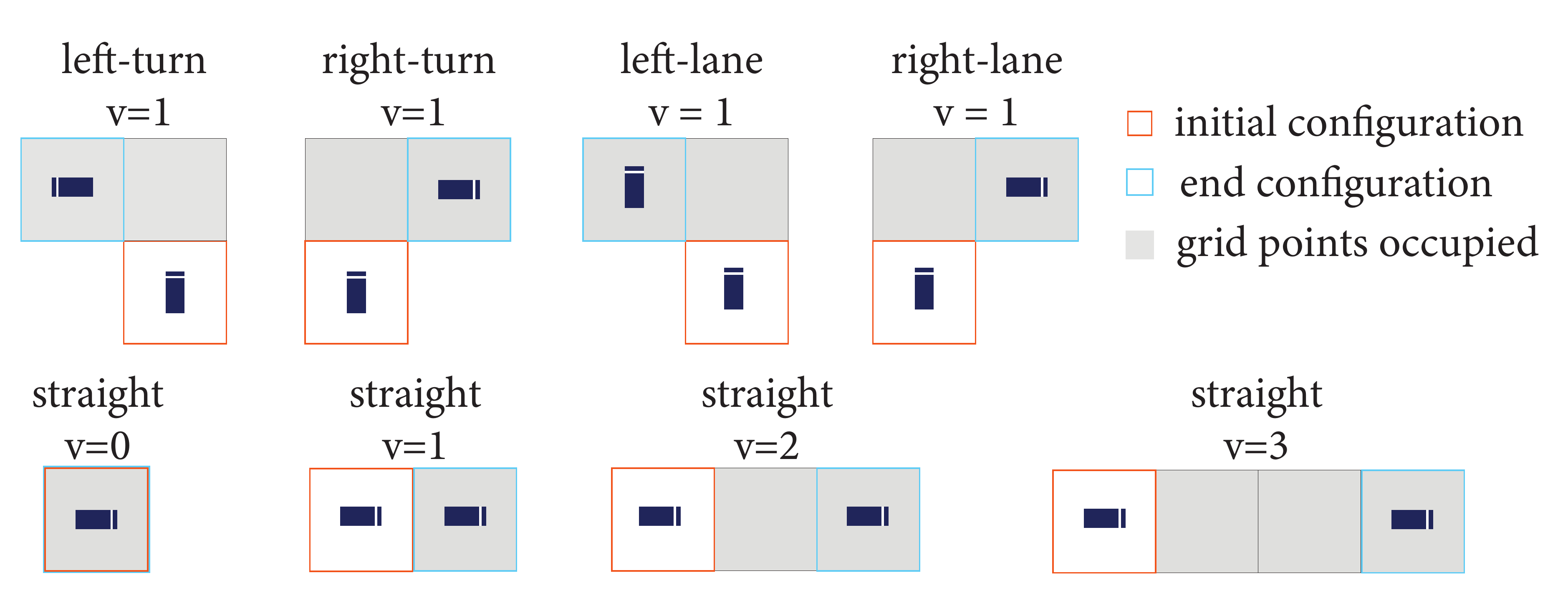}
\caption{Shows different grid point occupancy associated with different discrete agent maneuvers. Note the grid point occupancy represents a conservative space in which the agent may occupy when taking the associated maneuver.}
\label{fig_agent_dynamics}
\end{figure}

During an agent state transition, an agent may, depending on the maneuver, occupy a set of grid points. Before and after the state transition, the agent is assumed to occupy only a single grid point. Fig.~\ref{fig_agent_dynamics} shows the grid point occupancy for different agent maneuvers. The concept of grid point occupancy is defined as follows:
\begin{definition}[Grid Point Occupancy]
 The notion of grid point occupancy is captured by the definitions of the following maps for each $\text{Ag} \in \mathfrak{A}$. To define the grid point an agent is occupying at a given time we use the map: $\mathcal{G}_{\text{Ag},t}: S_{\text{Ag}} \rightarrow 2^{G}$, mapping each agent to the single grid point the agent occupies. By a slight abuse of notation, we let $\mathcal{G}_{Ag,t}: S_{\text{Ag}} \times {Act}_{\text{Ag}} \rightarrow 2^{G}$ be a function that maps each $s \in S_{\text{Ag}}$ and $a \in \rho_{\text{Ag}}(s)$ to denote the set of all grid points that are occupied by the agent $\text{Ag}$ when it takes an allowable action $a$ from state $s$ at the time-step $t$.
  \end{definition}

Here we assume that any graph-based planning algorithm can be used to specify an agent's motion plan, where the motion plan is a set of critical points along the graph that the agent must reach in order to get to its destination.

\subsection{Agent Backup Plan Action}
\label{backup_plan_action_def}
A \textit{backup plan} is a reserved set of actions an agent is \textit{entitled} to execute at any time while being immune to being at fault for a collision if one occurs. In other words, an agent will always be able to safely take its backup plan action. We show if each agent can maintain the ability to safely execute its own backup plan (i.e. keep a far enough distance behind a lead agent), the safety of the collective system safety is guaranteed. The default backup plan, which we refer to as $a_{\text{bp}}$ adopted here is that of applying maximal deceleration until a complete stop is achieved. Note, it may take multiple time-steps for an agent to come to a complete stop because of the inertial dynamics of the agent. 

\subsection{Limits on Agent Perception}
\label{limits_on_agent_perception}
In real-life, agents make decisions based on local information. We model this locality by defining a region of grid points around which agents have access to the full state and intentions of the other agents. We assume agents have different perception capabilities in different contexts of the road network. For road segments, the region around which agents make decisions cannot be arbitrarily defined. In fact, an agent's bubble must depend on its state, and the agent attributes and dynamics of all agents in the game. In particular, the bubble can be defined as follows:

\begin{definition}[Bubble]
\label{bubble_definition}
Let $\text{Ag}$ with state $s_0 \in S_{\text{Ag}}$. Then the bubble of $\text{Ag}$ with respect to agents of the same type is written as $\mathcal{B}_{\text{Ag}}(s_0)$. The bubble is the minimal region of space (set of grid points) agents need to have full information over to guarantee they can make a decision that will preserve safety under the defined protocol. 
\end{definition}

The details for the construction of the bubble for an agent with a particular set of attributes and dynamics can be found in the Appendix. At intersections, agents are assumed to be able to see across the intersection when making decisions about crossing the intersection. More precisely, any $\text{Ag}$ must be able to know about any $\text{Ag}' \in \mathfrak{A}$ that is in the lanes of oncoming traffic. The computation of the exact region of perception necessary depends on the agent dynamics. 

\section{Road Network Environment}
Here we introduce the structure of the road network environment that agents are assumed to be operating on. The road network is a grid world with additional structure (e.g. lanes, bundles, road segments, intersections, etc.). The road network is formalized as follows: 

\begin{definition}[Road Network]
A road network $\mathfrak{R}$ is a graph $\mathfrak{R} = (G, E)$ where $G$ is the set of grid points and $E$ is the set of edges that represent immediate adjacency in the Cartesian space among grid points. Note that each grid point $g\in G$ has a set of associated properties $\mathcal{P}$, where $\mathcal{P} = \{p, d, \texttt{lo} \}$ which denote the Cartesian coordinate, drivability of the grid point and the set of legal orientations allowed on the grid point respectively. Note, $p\in\mathbb{Z}^2$, $d \in \{0, 1\}$ and $\texttt{lo} \in \{\texttt{north}, \texttt{east}, \texttt{south}, \texttt{west}\}$.
\end{definition}

 $\mathcal{S}_{\text{sources}}$ ($\mathcal{S}_{\text{sinks}}$) are the set of grid points agents can enter or leave the road network from. Each intersection of the road network is governed by traffic lights. The road network is hierarchically decomposed into lanes, bundles and road segments, where a lane $La(g)$ defines a set of grid points that contains $g$ and all grid points that form a line going through $g$ and a bundle $Bu(g)$ is a set of grid points that make up a set of lanes that are adjacent or equal to the lane containing $g$ and have the same legal orientation. Each bundle can be decomposed into a set of road segments $RS$, where the intersections are used to partition each bundle into a set of road segments. These road components can be seen in Fig.\ref{fig_precedence}. 

We introduce the following graph definition since it will be used in the liveness proof. 
\begin{definition}[Road Network Dependency Graph]
\label{road_network_dependency_graph}
The road network dependency graph is a graph $G_{\text{dep}} = (RS, E)$ where nodes are road segments and a directed edge $(rs_1, rs_2)$ denotes that agents on $rs_1$ depends on the clearance of agents in $rs_2$ to make forward progress.
\end{definition}


\section{the Agent Protocol}
\label{section_agent_protocol}
The protocol is the set of rules agents use to select which action to ultimately take at a given time step. According to the protocol, agents first select an intended action using a profile. The protocol then defines additional rules that an agent uses to determine whether it has priority to take its intended action, and if not, which alternative, less-optimal actions it is allowed to take. The protocol is defined in a way that 1) scales well in the number of agents 2) is interpretable so there is a consistent and transparent way agents make their decisions 3) ensures safety and progress of all agents. In this section, we introduce the components that form the agent protocol that make it such that all these properties are satisfied. 



\subsection{Agent Precedence Assignment}
\label{section_precedence_rules}
The definition of the quasi-simultaneous game requires agents to locally assign precedence, i.e. have a set of rules to define how to establish which agents have higher, lower, equal or incomparable precedence to it.

Thus, the first element of the agent protocol is defining the agents' local precedence assignment algorithm so each agent knows its turn order relative to neighboring agents. Our precedence assignment algorithm is motivated by capturing how precedence among agents is generally established in real-life scenarios on a road network. In particular, since agents are designed to move in the forward direction, we aim to capture the natural inclination of agents to react to the actions of agents visibly ahead of it. 

Before presenting the precedence assignment rules, we must introduce a few definitions. Let us define: $\text{proj}^{B}_{\text{long}}: \mathfrak{A} \rightarrow \mathbb{Z}$, which is restricted to only be defined on the bundle $B$. In other words, $\text{proj}^{B}_{\text{long}}(Ag)$ is the mapping from an agent (and its state) to its scalar projection onto the longitudinal axis of the bundle $B$ the agent $\text{Ag}$ is in. If $\text{proj}^B_{\text{long}}(\text{Ag}') < \text{proj}^B_{\text{long}}(\text{Ag})$, then the agent $\text{Ag}'$ is behind $\text{Ag}$ in $B$.

The following rules can be used to define the precedence relation among agents $Ag$ and $Ag'$. 
\subsubsection{Local Precedence Assignment Rules}
\begin{enumerate}
    \item If $\text{proj}^B_{\text{long}}(\text{Ag}') < \text{proj}^B_{\text{long}}(\text{Ag})$ and $Bu(\text{Ag}')= Bu(\text{Ag})$, then $\text{Ag}' \prec \text{Ag}$, i.e. if agents are in the same bundle and $\text{Ag}$ is longitudinally ahead of $\text{Ag}'$, $\text{Ag}$ has higher precedence than $\text{Ag}'$.
    \item  If $\text{proj}^B_{\text{long}}(\text{Ag}') = \text{proj}^B_{\text{long}}(\text{Ag})$ and $Bu(\text{Ag}')= Bu(\text{Ag})$, then $\text{Ag} \sim Ag'$ and we say $\text{Ag}$ and $\text{Ag}'$ are equivalent in precedence.
    \item If $\text{Ag}'$ and $\text{Ag}$ are not in the same bundle, then the two agents are incomparable.
\end{enumerate}
Each agent $\text{Ag}\in \mathfrak{A}$ only assigns precedence according to the above rules locally to agents within its local region. Thus, we must show if all agents locally assign precedence according to these rules, a globally-consistent turn precedence among all agents is established. The linear ordering induced by these local rules are used to prove this. The reader is referred to the Appendix for the full proof.

\begin{figure}[H]
\centering
\includegraphics[scale=0.29]{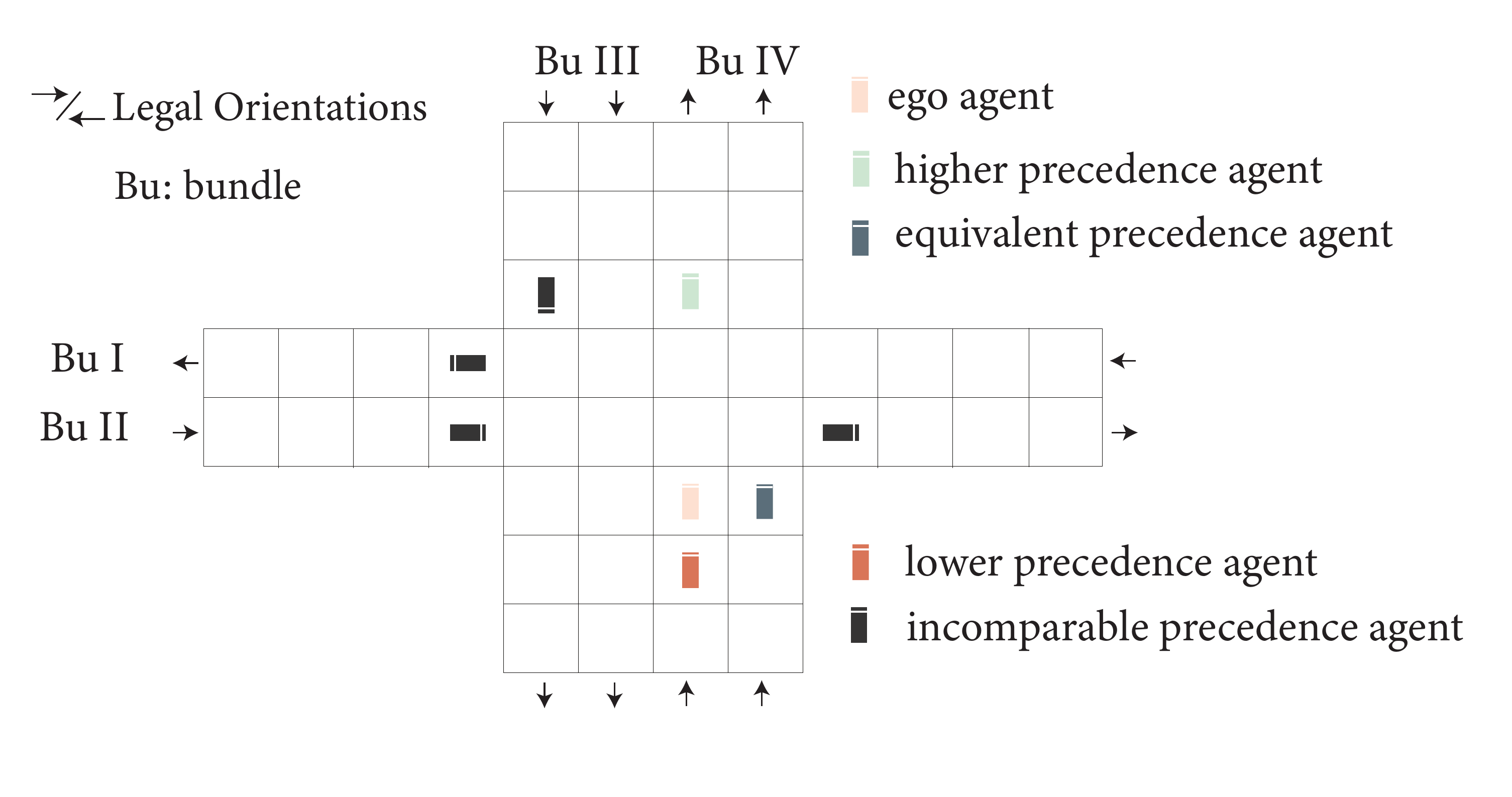}
\caption{Rules for precedence assignment.}
\label{fig_precedence}
\end{figure}


Even when this turn-order is established, there is still some ambiguity as to which agents have precedence. The ambiguity is resolved through the conflict-cluster resolution, introduced in Section \ref{section_conflict_cluster_resolution}. 

\subsection{Behavioral Profile}
\label{section_profiles}
The way in which agents select actions is the fundamental role of the agent protocol. The behavioral profile serves the purpose of defining which action an agent intends to take at a given time-step $t$. We define a specific assume-guarantee profile with the mathematical properties defined in \cite{agprofiles}. In particular, we define a set of ten different specifications (rules) and place a hierarchy of importance (ordering) on these rules. 

\begin{figure}[H]
\centering
\includegraphics[scale=0.16]{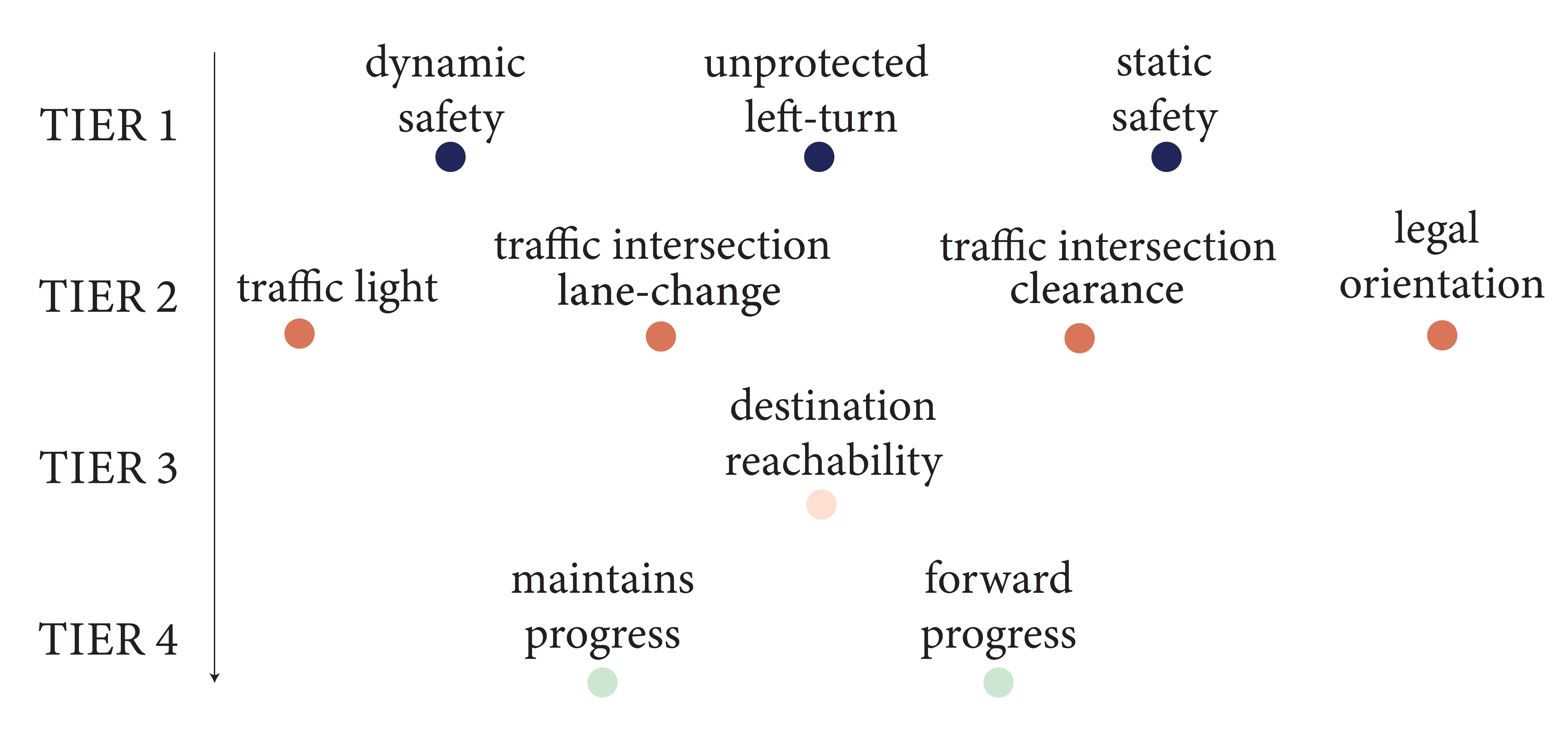}
\caption{Assume-guarantee profile that shows ordering of specifications, where specifications on the same tier are incomparable to one another and Tier 1 has highest priority.}
\label{fig_ag_profiles}
\end{figure}

Each of the specifications is associated with an oracle that evaluates whether or not an agent taking an action $a$ satisfies the specification. The reader is referred to the Appendix for the precise oracle definitions. The consistent-evaluating function, defined on this agent profile, will evaluate actions based on which subset of specifications they satisfy--giving priority to actions that satisfy the highest number of highest-valued specifications, as described in \cite{agprofiles}. The action with the highest value is then selected as the action the agent intends to takes.

For this work, the agent profile defined in Fig. \ref{fig_ag_profiles} is used to define both the agent's intended action $a_i$ and best straight action $a_{\text{st}}$ defined in Definition \ref{best_straight_action}. Since an agent would never propose a lane-change action if $O_{\text{Ag},t,\text{dynamic safety}}(s,a,u)$ were included in the profile, it is not included in the selection of the intended action $a_i$, but rather evaluated later downstream in the protocol. 


\subsection{Conflict-Cluster Resolution}
\label{section_conflict_cluster_resolution}
At every time-step $t$, each agent will know when to take its turn based on its local precedence assignment algorithm. Before taking its turn, the agent will have selected an intended action $a_i$ using the Agent profile. When it is the agent's turn to select an action, it must choose whether or not to take it's intended action $a_i$. When the intended actions of multiple agents conflict, the conflict-cluster resolution is a token-based querying method used to help agents determine which agent has priority in taking its action.

Under the assumption agents have access to the intentions of other agents within a local region as defined in Section \ref{limits_on_agent_perception}, agents can use the following criteria to define when it conflicts with another agent. 

\begin{definition}[Agent-Action Conflict]
\label{agent_action_conflict}
Let us consider an agent $\text{Ag}$ is currently at state $s \in S_{\text{Ag}}$ and wants to take action $a$ and an agent $\text{Ag}'$ at state $s' \in S_{\text{Ag}'}$ wants to take action $a'$. We write an agent-action conflict exists $(\text{Ag}, s, a) \dagger (\text{Ag}', s', a')$, if each of the agents taking their respective actions will cause them to overlap in occupancy grid points or end up in a configuration where the agent behind does not have a valid safe backup plan action. 
\end{definition}

 In the case that an agent's action is in conflict with another agents' action, the agent must send a conflict request that ultimately serves as a bid the agent is making to take its intended action. It cannot, however, send requests to just any agent (e.g. agents in front of it). The following criteria are used to determine the properties that must hold in order for an agent $\text{Ag}$ to send a conflict request to agent $\text{Ag}'$: 1) $\text{Ag}$'s intended action $a_i$ is a lane-change action, 2) $\text{Ag}' \in \mathcal{B}_{Ag}(s)$, i.e. $\text{Ag}'$ is in agent $\text{Ag}$'s bubble, 3) $\text{Ag}' \precsim Ag$, i.e. $\text{Ag}$ has equivalent or higher precedence than $\text{Ag}'$, 4) $Ag$ and $Ag'$ have the same heading, 5) $(Ag, a_i) \dagger (Ag',a_i')$: agents intended actions are in conflict with one another,  and 6) $\mathcal{F}_{\text{Ag}}(u,a_i) = \texttt{F}$, where $\mathcal{F}_{\text{Ag}}(u,a_i)$ is the max-yielding-not enough flag and is defined below. 

\begin{definition}[maximum-yielding-not-enough flag]
\label{definition_max_yielding_flag_not_enough}
The maximum-yielding-not-enough flag $\mathcal{F}_{\text{Ag}}:\mathcal{U} \times Act_{\text{Ag}} \rightarrow \mathbb{B}$ is set to $\texttt{T}$ when $\text{Ag}$ is in a configuration where if $\text{Ag}$ did a lane-change, $Ag$ would still violate the safety of $\text{Ag}'$'s backup plan action even if $\text{Ag'}$ applied its own backup plan action.
\end{definition}
We note that if $\mathcal{F}_{\text{Ag}}(u, a_i)$ is set, $\text{Ag}$ cannot send a conflict request by the last condition. Even though $\text{Ag}$ does not send a request, it must use the information that the flag has been set in the agent's Action Selection Strategy defined in Section \ref{section_action_selection_strategy}. After a complete exchange of conflict requests, each agent will be a part of a cluster of agents that define the set of agents it is ultimately bidding for its priority (to take its intended action) over. These clusters of agents are defined as follows: 
\begin{definition}[Conflict Cluster]
\label{conflict_cluster_definition}
A conflict cluster for an agent $\text{Ag}$ is defined as $\mathcal{C}_{Ag} = \{Ag' \in \mathfrak{A} \mid Ag \texttt{ send } \text{Ag}' \text{ or } Ag' \texttt{ send } Ag\}$, where $\text{Ag} \texttt{ send } Ag'$ implies $\text{Ag}$ has sent a conflict request to $\text{Ag}'$. An agents' conflict cluster defines the set of agents in its bubble that an agent is in conflict with.
\end{definition}

Fig. \ref{fig_conflict_clusters} shows an example scenario and each agents' conflict clusters. Once the conflict requests have been sent and an agent can thereby identify the other agents in its conflict cluster, it needs to establish whether or not the conflict resolution has resolved in it's favor. 

\begin{figure}[H]
\centering
\includegraphics[scale=0.30]{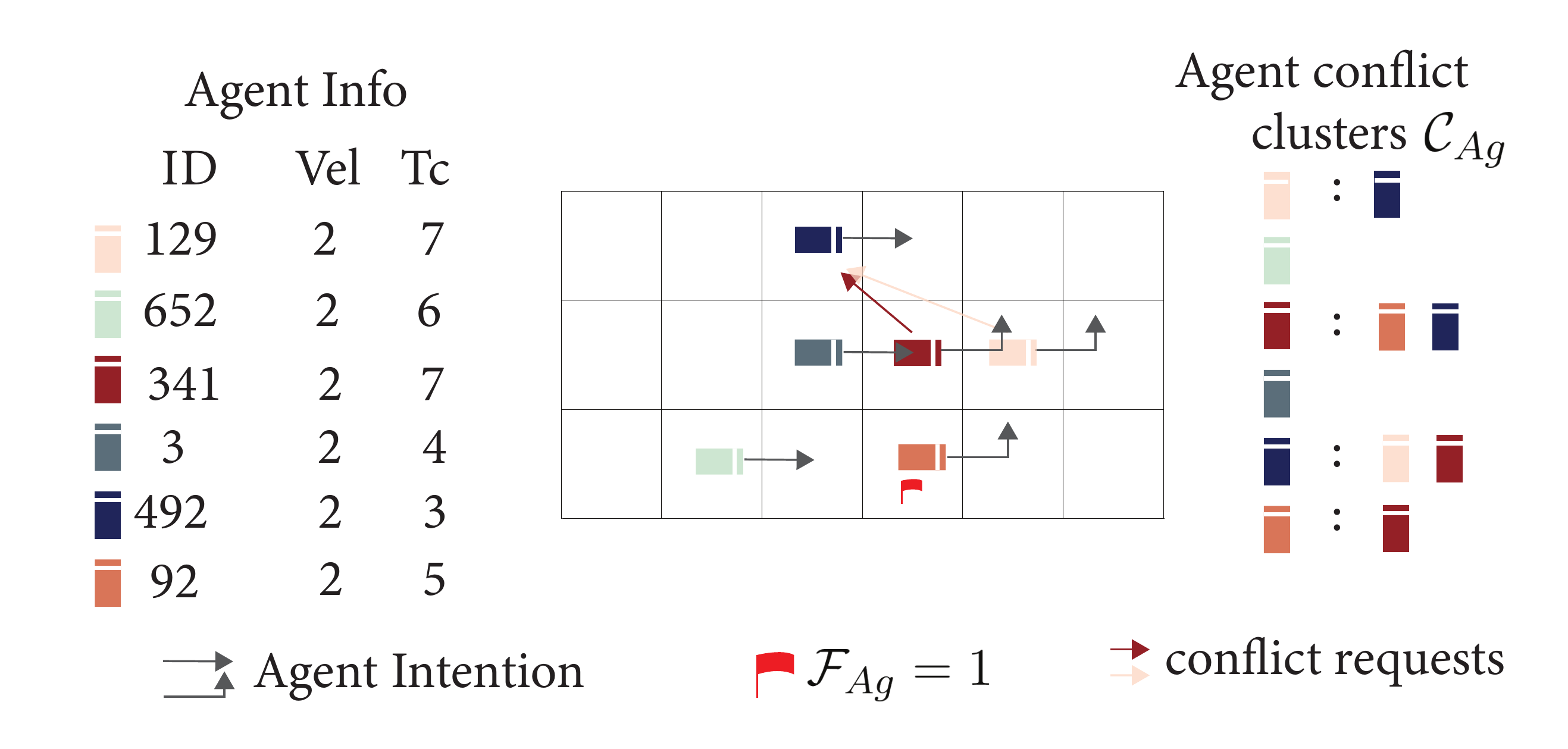}
\caption{An example scenario with agents in a given configuration of agents, their intended actions and their respective conflict clusters. }
\label{fig_conflict_clusters}
\end{figure}

Once an agent has determined which agents are in its conflict cluster, it must determine whether or not it has the priority to take its intended action. The token resolution scheme is the way in which agents determine whether they have precedence. 

The token resolution strategy must be designed to be fair, meaning each agent will always eventually wins their conflict resolution. The resolution is therefore based on the agents' token counts $\texttt{Tc}$, which is updated by agents to represent how many times an agent has been unable to take a forward progress action thus far. 

The token count updates according to the agent's chosen action. In particular, if $\text{Ag}$ selects action $a$: if $O_{\text{forward progress}}(s,a,u) = \texttt{T}$, the the token count resets to 0, otherwise it increases by 1. 

Then, a fair strategy would be to make it so that the agent with the highest amount of tokens wins in its own conflict cluster. Thus, we define a token resolution indicator variable for each $\text{Ag}$ as $\mathcal{W}_{\text{Ag}} \in \mathbb{B}$, indicating whether or not the agent has won in its conflict cluster. 
The conflict cluster resolution indicator variable $\mathcal{W}_{\text{Ag}}$ evaluates to $\texttt{T}$ if $Ag$ has the highest amount of tokens in its conflict cluster, where ties are broken via agent ID comparison.

\subsection{Action Selection Strategy}
\label{section_action_selection_strategy}
The Action Selection Strategy is a decision tree that defines whether or not an agent is allowed to take its intended action $a_i$ and if it is not, which alternative action it should take. In the case where an agent is not allowed to take $a_i$, the agent is restricted to take either: the best straight action $a_{st}$, which is defined in Definition \ref{best_straight_action}, or its backup plan action $a_{bp}$, where the best straight action is defined as follows: 
\begin{definition}[Best Straight Action]
\label{best_straight_action}
Let us consider $\text{Ag}$ and its associated action set $\rho_{\text{Ag}}(s)$. The best straight action is the action $a \in \rho_{\text{Ag}(s)}$ that is the highest-ranked action (according to the profile defined in Section \ref{section_profiles}), among the set of all actions for which $\gamma_{\text{Ag}} = \text{straight}$. 
\end{definition}

The decision tree branches are defined based on the following five conditions: 1) $a_i$, the agent's and other agents' (in its bubble) intended actions 2) $\text{Ag}$'s role in conflict request cluster being a) a conflict request sender, b) a conflict request receiver, c) both a sender and a receiver, d) neither sender or receiver, 3) the agent's conflict cluster resolution $\mathcal{W}_{\text{Ag}}$, 4) evaluation of $O_{\text{Ag}, t, \text{dynamic safety}}(s,a_i,u)$ and 5) $\mathcal{F}_{\text{Ag}}(u, a_i)$ for $\text{Ag}$ is raised, where $\mathcal{F}_{\text{Ag}}(u, a_i)$ is the maximal-yielding-not-enough flag defined in Section \ref{section_conflict_cluster_resolution}. 

If an agent receives a conflict cluster request and loses their conflict cluster resolution, according to the action selection strategy, the agent must take its backup plan action $a_{\text{bp}}$. An agent is only allowed to take a lane-change action when the agent is a winner of its conflict cluster resolution, $\mathcal{F}_{\text{Ag}} = \texttt{F}$ and the dynamic safety oracle evaluates to true (i.e. $O_{\text{Ag}, t, \text{dynamic safety}}(s,a_i,u) = \texttt{T}$). Finally, an agent that loses in its conflict cluster but did not send requests must take $a_{\text{st}}$. A figure showing the full decision-tree logic for selecting actions can be found in the Appendix. 

The agent protocol, as described in the above sections, has been designed in a way such that if all agents are selecting actions via the protocol, we can provide formal guarantees on safety and liveness. Theses safety and liveness proofs are given in the following sections. 

\section{Formal Guarantees}
\label{all_assumptions}
Before introducing the formal guarantees of safety and liveness and their respective proofs, we first make explicit the assumptions that must hold on agents and the road network. 
\begin{enumerate}
\item Each $\text{Ag} \in \mathfrak{A}$ has access to the traffic light states.
\item There is no communication error in the conflict requests, token count queries and the agent intention signals.
\item All intersections in the road network $R$ are governed by traffic lights.
\item The traffic lights are designed to coordinate traffic such that if agents respect the traffic light rules, they will not collide.
\item Agents follow the agent dynamics defined in Section \ref{agent_attributes}.
\item For $t=0$, $\forall \text{Ag} \in \mathfrak{A}$ in the quasi-simultaneous game is initialized to be located on a distinct grid point on the road network and have a safe backup plan action $a_{bp}$ such that $S_{\text{Ag}, bp}(s, u) = \texttt{T}$.
\item The traffic lights are red a window of time $\Delta t_{\text{tl}}$ such that $t_{\text{min}}<\Delta t_{\text{tl}}<\infty$, where $t_{\text{min}}$ is defined so agents are slowed down long enough so agents that have been waiting can take a lane-change action. More details can be found in the Appendix.
\item The static obstacles are not on any grid point $g$ where $g.d = 1$.
\item Each $\text{Ag}$ treats its respective goal $\text{Ag}.\texttt{g}$ as a static obstacle.
\item Bundles in the road network $\mathfrak{R}$ have no more than 2 lanes.
\item All intersections in the road network $\mathfrak{R}$ are governed by traffic lights. 
\end{enumerate}

\subsection{Safety Guarantee}
\label{section_safety_proof}
Safety is guaranteed when agents do not collide with one another. An agent causes collision when it takes an action that satisfies the following condition.
\begin{definition}[Collision]
\label{definition_collision}
An agent $\text{Ag}$ that takes an action $a \in Act_{\text{Ag}}$ will cause collision if the grid point occupancy of $Ag$ ever overlaps with the grid point occupancy of another agent $Ag'$ or a static obstacle $O_{\text{st}}$.
\end{definition} 

A strategy where agents simply take actions that avoid collision in the current time-step is insufficient for guaranteeing safety because of the inertial properties of the agent dynamics. The agent protocol has therefore been defined so an agent also avoids violating the safety of its own and any other agent's backup plan action $a_{bp}$ defined in Section \ref{backup_plan_action_def}. An agent's backup plan action $a_{bp}$ is evaluated to be safe when the following conditions hold: 
\begin{definition}
\label{definition_safety_back_up_plan_action}[Safety of a Backup Plan Action] Let us define the safety of an agent's backup plan action $S_{\text{Ag}, bp}:\mathcal{U} = \mathbb{B}$, where $\mathbb{B} = \{\texttt{T}, \texttt{F} \}$ is an indicator variable that determines whether an agent's backup plan action is safe or not. It is defined as: $S_{Ag, bp}(u) = \land_{o \in O} o(s, a_{bp}, u)$ where the set $O$ is the set of all oracles in the top three tiers of the agent profile defined in Section \ref{section_profiles}.
\end{definition}

An agent $\text{Ag}$ takes an action $a \in Act_{\text{Ag}}$ that violates the safety backup plan action of another agent $\text{Ag}'$ when the following conditions hold:
\begin{definition}[Safety Backup Plan Violation Action]
\label{safety_backup_plan_violation_action}
Let us consider an agent $\text{Ag}$ that is taking an action $a \in Act_{\text{Ag}}$, and another agent $\text{Ag}'$. The action $(\text{Ag},a) \bot \text{Ag}'$, i.e. agent $\text{Ag}$ violates the safety backup plan of an agent $\text{Ag}'$ when by taking an action $a$, then $S_{\text{Ag}', bp}(u') = \texttt{F}$, where $u'$ is the state of the game after $\text{Ag}$ has taken its action. In other words, by taking the action, the agent has ended in a state such that it violates the safety of its own or another agents' backup plan action.

\end{definition}
The safety proof is based on the premise that all agents only take actions that do not collide with other agents and maintain the invariance of the safety of their own \textit{and} other agents' safety backup plan actions. The safety theorem statement and the proof sketch are as follows.

We can treat the quasi-simultaneous game as a program, where each of the agents are separate concurrent processes. A safety property for a program has the form $P \Rightarrow \square Q$, where $P$ and $Q$ are immediate assertions. This means if the program starts with $P$ true, then $Q$ is always true throughout its execution \cite{owicki1982proving}.

\begin{theorem}[Safety Guarantee]
\label{safety_guarantee_theorem}
Given all agents $\text{Ag} \in \mathfrak{A}$ in the quasi-simultaneous game select actions in accordance to the Agent Protocol specified in Section \ref{section_agent_protocol}, then we can show the safety property $P \Rightarrow \square Q$, where the assertion $P$ is an assertion that the state of the game is such that $\forall Ag, S_{\text{Ag}, bp}(s,u) = \texttt{T}$, i.e. each agent has a backup plan action that is safe, as defined in Section \ref{definition_safety_back_up_plan_action}. We denote $P_t$ as the assertion over the state of the game at the beginning of the time-step $t$, before agents take their respective actions. $Q_t$ is the assertion that the agents never occupy the same grid point when taking their respective action at time step $t$. 
\end{theorem} 
The following is a proof sketch.

\begin{proof}
To prove an assertion of this form, we need to find an invariant assertion $I$ for which i) $P \Rightarrow I$ ii) $I \Rightarrow \square I$ and iii) $I \Rightarrow Q$ hold.
 We define $I$ to be the assertion that holds on the actions that agents select to take at a time-step. We denote $I_t$ to be the assertion on the actions agents take at time $t$ such that $\forall Ag$, $\text{Ag}$ takes $a\in Act_{\text{Ag}}$ where 1) it does not collide with other agents and 2) it does not violate the safety of other agents' back up plan actions (i.e. $\forall \text{Ag}, S_{\text{Ag}, bp}(u') = \texttt{T}$ where $s' = \tau_{\text{Ag}}(s, a)$, and $u'$ is the corresponding global state of the game after each $\text{Ag}$ has taken its respective action $a$).


We can prove $P \Rightarrow \square Q$ by showing the following:
\begin{enumerate}
    \item $P_t \Rightarrow I_t$. This is equivalent to showing that if all agents are in a state where $P$ is satisfied at time $t$, then all agents will take actions at time $t$ where the $I$ holds. This can be proven by showing agents will take actions that satisfy the conditions of $I$ as long as they are begin a state where all agents have a safe backup action and they select actions according to the protocol.
    \label{p_implies_i}
    \item $I \Rightarrow \square I$. If agents take actions such that at time $t$ such that the assertion $I_t$ holds, then by the definition of the assertion $I$, agents will end up in a state where at time t+1, assertion $P$ holds, meaning $I_t \Rightarrow P_{t+1}$. Since $P_{t+1} \Rightarrow I_{t+1}$ from \ref{p_implies_i}, we get $I \Rightarrow \square I$.
    \item  $I \Rightarrow Q$. If all agents take actions according to the assertions in $I$, then collisions will not occur. This follows from the definition of $I$. 
\end{enumerate}
\end{proof}

The reader is referred to the Appendix for a full proof. Proof of safety alone is not sufficient reason to argue for the effectiveness of the protocol, as all agents could simply stop for all time and safety would be guaranteed. A liveness guarantee, i.e. proof that all agents will eventually make it to their final destination, is critical. In the following section, we present liveness guarantees.

\subsection{Liveness Guarantees}
\label{section_liveness_proof}
A liveness property asserts that program execution eventually reaches some desirable state \cite{owicki1982proving}. In this paper, we describe the eventual desirable state for each agent is to reach their respective final destinations. Unfortunately, deadlock occurs when agents indefinitely wait for resources held by other agents \cite{deadlock2010}. Since the Manhattan grid road network has loops, agents can enter a configuration in which each agent in the loop is indefinitely waiting for a resource held by another agent. When the density of agents in the road network is high enough, deadlocks along these loops will occur. We can therefore guarantee liveness only when certain assumptions hold on the density of the road network.

\begin{definition}[Sparse Traffic Conditions]
\label{assumptions_sparse_traffic}
Let $M$ denote the number of grid points in the smallest loop (defined by legal orientation) of the road network, not including grid points $g\in\mathcal{S}_{\text{intersections}}$. The sparsity condition must be such that $N < M-1$, where $N$ is the number of agents in the road network. The number of agents has to be such that the smallest loop does not become completely saturated, in which deadlock would occur. Note, these sparsity conditions are conservative because it is a bound defined by the worst possible assignment of agents and their destinations.
\end{definition}

Now, we introduce the liveness guarantees under these sparse traffic conditions. The proof of liveness is based on the fact that 1) agent profile include progress specifications and 2) conflict precedence is resolved by giving priority to the agent that has waited the longest time (a quantity that is reflected by token counts). 

\begin{theorem}[Liveness Under Sparse Traffic Conditions]
\label{liveness_guarantee_theorem}
Under the Sparse Traffic Assumption given by Definition  \ref{assumptions_sparse_traffic} and given all agents $\text{Ag} \in \mathfrak{A}$ in the quasi-simultaneous game select actions in accordance to the Agent Protocol specified in Section \ref{section_agent_protocol}, liveness is guaranteed, i.e. all $\text{Ag} \in \mathfrak{A}$ will always eventually reach their respective goals.
\end{theorem}

The following is a proof sketch.
\\
\begin{proof}
\begin{enumerate}
    \item The invariance of a no-deadlock state follows from the sparsity assumption and the invariance of safety (no collision) follows from the safety proof. 
    \label{invariance_no_deadlock}
    \item Inductive arguments related to control flow are used to show that all $\text{Ag}$ will always eventually take $a \in Act_{\text{Ag}}$ where $O_{\text{forward progress}}(s,a,u) = \texttt{T}$.
    \label{last_liveness_item}
    \begin{enumerate}
        \item Let us consider a road segment $r \in RS$ that contains grid point(s) $g \in \mathcal{S}_{\text{sinks}}$ meaning that the road segment contains grid points with sink nodes. Inductive arguments based on the agents' longitudinal distance to destination grid points are used to show every $\text{Ag} \in r$ will be able to always eventually take $a \in Act_{\text{Ag}}$ for which the forward progress oracle $O_{\text{forward progress}}(s,a,u) = \texttt{T}$.
        \label{liveness_lemma_1}
        \item Let us consider a road segment $rs \in RS$. Let us assume $\forall rs \in RS, \exists (rs, rs') \in G_{\text{dep}}$ meaning that the clearance of $rs$ depends on the clearance of all $rs'$.  Inductive arguments based on agents' longitudinal distance to the front of the intersection show any $\text{Ag}$ on $rs$ will always eventually take $a \in Act_{\text{Ag}}$ where the forward progress oracle $O_{\text{forward progress}}(s,a,u) = \texttt{T}$.
        \label{liveness_lemma_2}
        \item For any $\mathfrak{R}$ where the dependency graph $G_{\text{dep}}$ (as defined in Definition  \ref{road_network_dependency_graph}) is a directed-acyclic-graph (DAG), inductive arguments based on the linear ordering of road segments $rs \in G_{\text{dep}}$, combined with the arguments \ref{liveness_lemma_1}-\ref{liveness_lemma_2}, can be used to prove all $\text{Ag} \in \mathfrak{A}$ will always eventually take $a\in Act_{\text{Ag}}$ for which the forward progress oracle $O_{\text{forward progress}}(s,a,u)=\texttt{T}$.
        \label{liveness_induction_statement}
        \item When the graph $G_{\text{dep}}$ is cyclic, the Sparsity Assumption \ref{assumptions_sparse_traffic} allows for similar induction arguments in \ref{liveness_induction_statement} to apply.
    \end{enumerate}
    \item By the above inductive arguments and the definition of $O_{\text{forward progress}}(s,a,u)$, all $\text{Ag}$ will always eventually take actions that allow them to make progress towards their respective destinations.
\end{enumerate}
\end{proof}
\noindent
 The reader is referred to the Appendix for a full proof.

\section{Simulation Results}
In order to streamline discrete-time multi-agent simulations, we have built a traffic game simulation platform called Road Scenario Emulator (RoSE). We use RoSE to generate different game scenarios and simulate how agents will all behave if they each follow the agent strategy protocol introduced in this paper. 

\begin{figure}[H]
\centering
\includegraphics[scale=0.34]{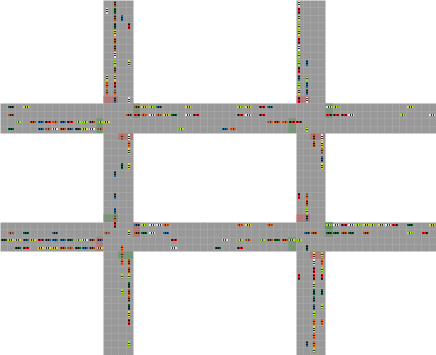}
\caption{City blocks map environment.}
\label{fig_city_blocks_med}
\end{figure}

We simulate the game with randomized initialization of spawning agents at the source nodes for three different road network environments: 1) the straight road segment, 2) small city blocks grid and 3) large city blocks grid. A snapshot of a small city blocks grid simulation is shown in Fig. \ref{fig_city_blocks_med}. 

The agent attributes in this simulation are as follows: $v_{\text{min}}=0$, $v_{\text{max}}=3$, $a_{\text{min}}=-1$, and $a_{\text{max}}=1$. For each road network environment, we simulate the game 100 times for $t=250$ time-steps. During each time-step, agents will spontaneously spawn with some defined probability $p$ at the source nodes and are randomly assigned a sink node as their destination. Agents that make it to their destinations exit the map. For all game simulation trials, collision does not occur. Although liveness is only guaranteed in sparse traffic conditions, we simulate for a number of agents $N > M-1$ specified in the sparsity condition and agents do not enter a deadlock state. In particular, over the 100 trials for each of the maps (straight, small and large city blocks), on average $77\%$, $36\%$ and $43\%$ made it to their respective destinations on the respective maps by the end of the 250 time-steps.

\section{Conclusion and Future Works}
In this paper, we have proposed a novel paradigm for designing safety-critical decision-making modules for agents whose behavior is extremely complex and highly-coupled with other agents. The main distinction of our proposed architecture from the existing literature, is the shift from thinking of each agents as separate, individual entities, to agents as a collective where all \textit{all} agents adopt a \textit{common} local, decentralized protocol. The protocol defines the agent attributes, the region it must reason over (i.e. the bubble), how the agent chooses its intended agent, and how it ultimately selects which action to take. With this protocol, we are able to formally guarantee specifications safety and liveness (under sparse traffic conditions) for all agents. We validate the safety and liveness guarantees in a randomized simulation environment. 

The current work still lacks 1) liveness guarantees in all scenarios, 2) robustness to imperfect sensory information and 3) does not account for other agent types like pedestrians and cyclists. Future work on modifying the agent strategy architecture to prevent the occurrence of the loop deadlock introduced in Section \ref{section_liveness_proof} from occurring. Additionally, the architecture must be modified in a way to effectively accommodate impartial and imperfect information. We also hope to accommodate a diverse, heterogenous set of car agents and also other agent types like pedestrians and cyclists. Although the work needs to be extended to make more applicable to real-life systems, we believe this work is a first step towards defining a comprehensive method for guaranteeing safety and liveness for all agents in an extremely dynamic and complex environment.

\section*{Acknowledgments}
We would like to acknowledge K. Mani Chandy who provided valuable input and to Giovanna Amorim for her contributions to the simulation code.

 \section*{Author Contributions}
 K.X.C., R.M.M., and T.P-M. jointly conceived the conceptual framework. K.X.C. and T.P-M. jointly developed the problem formulation and theoretical approach. K.X.C. worked out the main proofs with input from T.P-M. K.X.C. drafted the manuscript and figures with input from T.P-M. S-J.C. and R.M.M. provided guidance on the overall approach and provided feedback on the final manuscript.

\bibliographystyle{plain}
\bibliography{refs}
\nocite{*}

\newpage
\clearpage
\newpage

\begin{appendices}
\label{appendix}

\subsection{Road Network}
The following defines the set of properties that grid points can have.
\subsubsection{Grid Point Properties}
The set of properties $\mathcal{P} = \{p, d,$ \texttt{lo}$\}$ of each grid point $g \in G$. $p \in \mathbb{Z}^2$ denotes the Cartesian coordinate of the grid point, $d \in \{0, 1\}$, which is an indicator variale that defines whether or not the grid point is drivable, \texttt{lo} is the legal orientation, where the legal orientation is an element of the set $ \{\texttt{north}, \texttt{east}, \texttt{south}, \texttt{west}\}$. The set \texttt{lo} may be empty when the grid point is not drivable.

\begin{figure}[H]
\centering
\includegraphics[scale=0.20]{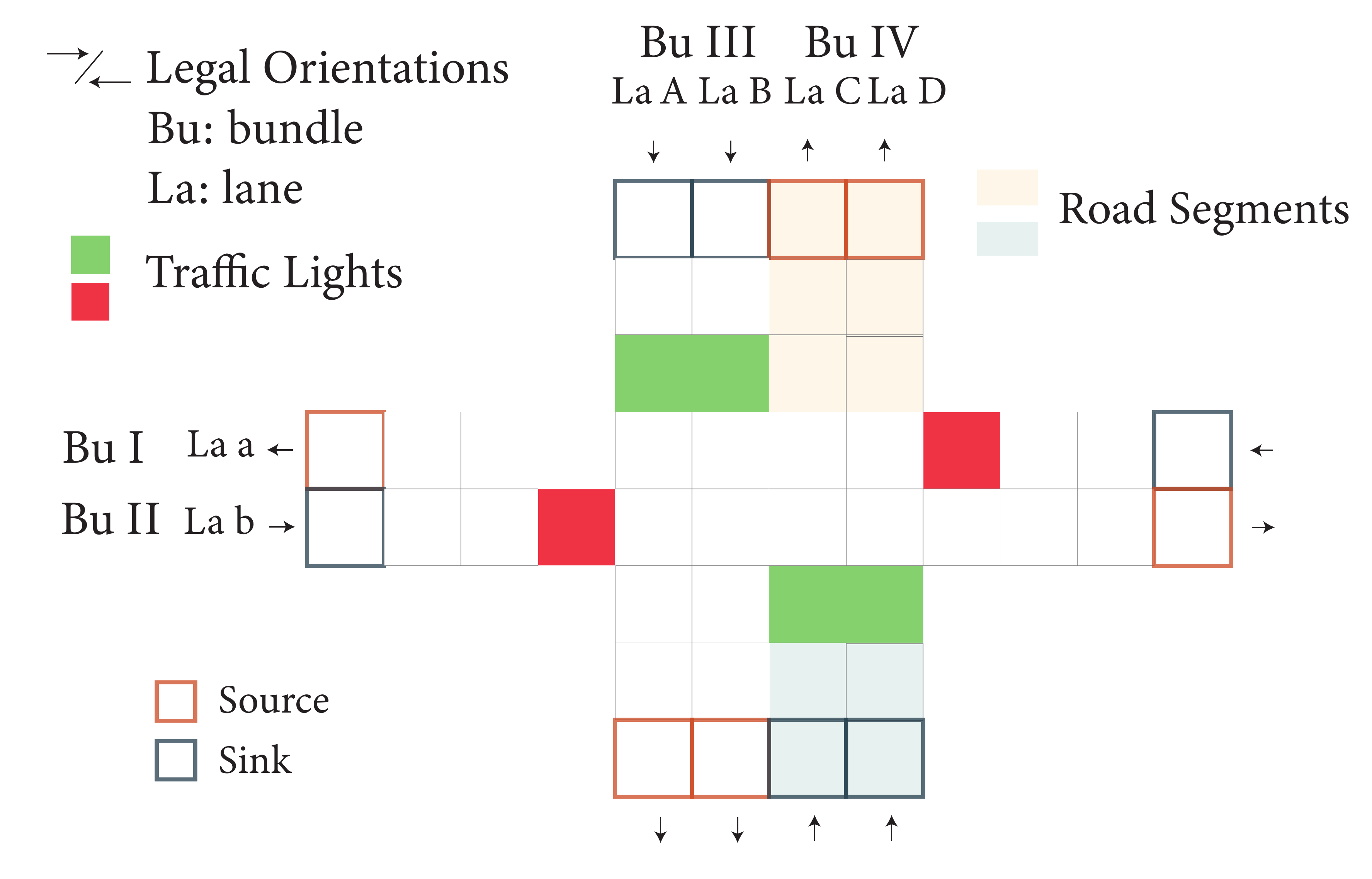}
\caption{Road network decomposition where each box represents a grid point.}
\label{fig_road_network}
\end{figure}

The following are sets of labeled grid points on the road network map. 
\begin{enumerate}
    \item $\mathcal{S}_{\text{intersection}}$: A set of grid points that contains all grid points with more than one legal orientation.
    \item $\mathcal{S}_{\text{traffic light}}$: A set of grid points that represent the traffic light states in the vertical or horizontal direction via its color (for every intersection). 
\end{enumerate}

The road network is hierarchically decomposed into lanes and bundles, which are defined informally as follows: 
\begin{itemize}
    \item Lanes: Let lane $La(g)$ denote a set of grid points that contains all grid points that are in the same `lane' as $g$. $La(g) = \{g' | \text{proj}_x(g'.p) = \text{proj}_x(g.p)$ or $\text{proj}_y(g.p) = \text{proj}_y(g.p),$ \\ 
    $g'.\phi_l = g.\phi_l, g.\texttt{drivable} = g'.\text{drivable} = 1 \}$. 
    \item Bundles: First, we define the set of adjacent lanes to lane $La(g)$ as $\texttt{adj}(La(g)) = \{La(g') \mid \exists e = (\hat{g}, \hat{g}') \in \mathfrak{R} \text{ s.t. } (\hat{g} \in La(g), \hat{g'} \in La(g')) \text{ and } \hat{g}.\phi_l = \hat{g'}.\phi_l\}$. This represents the set of lanes $La(g)$ in the same direction that the lane is adjacent to. 
    Let $N(g)$ = \texttt{adj}(La(g)). Let bundle $Bu(g)$ denote a set of lanes that are all connected to one another and is defined recursively as follows:
    \[
    Bu(g) =
    \begin{cases}
        La(g) \cup N(g) \text{ if } N(g) \neq \emptyset  \\
        La(g) & \text{otherwise}.
    \end{cases}
    \]
\end{itemize}


For clarity of the road network decomposition, refer to Fig. \ref{fig_precedence}. With slight abuse of notation, we let $La(\text{Ag})$ refer to the lane ID associated with the grid point $(s.x_{\text{Ag}}, s.y_{\text{Ag}})$, and $Bu(\text{Ag})$ mean the bundle ID associated with the lane $La(\text{Ag})$.

\subsection{Agent Backup Plan Action}
\begin{definition}[Backup Plan Action]
\label{def:backup_plan_action}
The backup plan action $a_{bp}$ is a control action where $a = a_{\text{min}}$ and when applying $a_{\text{min}}$ causes the agent's velocity to go below 0, $a=\text{max}(a_{\text{min}}, -s.v_{Ag})$ and $\gamma_{\text{Ag}} = \texttt{straight}$. 
\end{definition}

\subsection{Bubble Construction}
\label{bubble_construction}

In order to define the bubble for the agent dynamics specified in Section \ref{agent_attributes}, we present some preliminary definitions. We first introduce the backup plan node set (which is defined recursively) as follows:
\begin{definition}[Backup Plan Node Set]
Let $\text{Ag} \in \mathfrak{A}$ and $s_0 \in S_{\text{Ag}}$.
The backup plan grid point set $BP_{\text{Ag}}(s_0)$ is all the grid points agent $\text{Ag}$ occupies as it applies maximum deceleration to come to a complete stop.
\[
  BP_{\text{Ag}}(s_0) =
  \begin{cases}
    \mathcal{G}_{\text{Ag}}(s_0, a_{{\text{bp}}}) \cup BP_{\text{Ag}}(\tau_{\text{Ag}}(s_0,a_{{\text{bp}}})) & \text{if $\tau_{\text{Ag}}(s_0, a_{{\text{bp}}}).v \neq 0$} \\
   \mathcal{G}_{\text{Ag}}(\tau(s_0, a_{{\text{bp}}})) & \text{otherwise}.
  \end{cases}
\]
where $a_{\text{min}}$ is the agent's action of applying maximal deceleration while keeping the steering wheel at the neutral position.
\end{definition}

\begin{definition}[Forward/Backward Reachable States] The (1-step) forward reachable state set of agent $\text{Ag}$ denoted $\mathcal{R}_{\text{Ag}}(s_0)$ represents the set of all states reachable by $\text{Ag}$ from the state $s_0$. The forward reachable set is defined as $\mathcal{R}_{\text{Ag}}(s_0) \triangleq \{ s \in S_{\text{Ag}} \mid \exists a \in \rho_{\text{Ag}}(s_0).s = \tau(s_0, a)\}$.
Similarly, we define the (1-step) backward reachable state set $\mathcal{R}^{-1}_{\text{Ag}}(s_0)$ as the set of all states from which the state $s_0$ can be reached by $\text{Ag}$. Formally, $\mathcal{R}^{-1}_{\text{Ag}}(s_0) \triangleq \{s \in S_{\text{Ag}} \mid \exists s \in S_{\text{Ag}}. \exists a \in \rho_{\text{Ag}}(s). s_0 = \tau(s, a) \}.$
\end{definition}

\begin{definition}[Forward Reachable Nodes]
We denote by $\mathcal{G}^{\mathcal{R}}_{\text{Ag}}(s_0)$ the \textit{forward reachable node set}, namely, the set of all grid points that can be occupied upon taking the actions that brings the agent $\text{Ag}$ from its current state $s_0$ to a state in $\mathcal{R}_{\text{Ag}}(s_0)$. Specifically, 
$$\mathcal{G}^{\mathcal{R}}_{\text{Ag}}(s_0) \triangleq \bigcup_{a \in \rho_{Ag}(s_0)} \mathcal{G}_{\text{Ag}}(s_0, a)$$ 

\end{definition}
This set represents all the possible grid points that can be occupied by an agent in the next time step.

\begin{definition}[Occupancy Preimage]
For $n \in G$, where $G$ are the nodes in the road network graph $\mathfrak{R}$, the \textit{occupancy preimage} $\mathcal{G}^{\mathcal{R}^{-1}}_{\text{Ag}}(n) $ is the set of states of agent $\text{Ag}$ from which there is an action that causes $n$ to be occupied in the next time step. Formally, 
$$\mathcal{G}^{\mathcal{R}^{-1}}_{\text{Ag}}(n) = \{s \in S_{\text{Ag}} \mid \exists a \in \rho_{\text{Ag}}(s). n \in \mathcal{G}_{\text{Ag}}(s, a)\}$$ 
\end{definition}

In the next section, we define several different sets of grid points that are defined to represent the locations where two agents may possibly interfere with one another, which are shown in Fig. \ref{fig_build_bubbles}. The bubble is defined to be the union of these sets of grid points. 

\begin{figure}[H]
\centering
\includegraphics[scale=0.18]{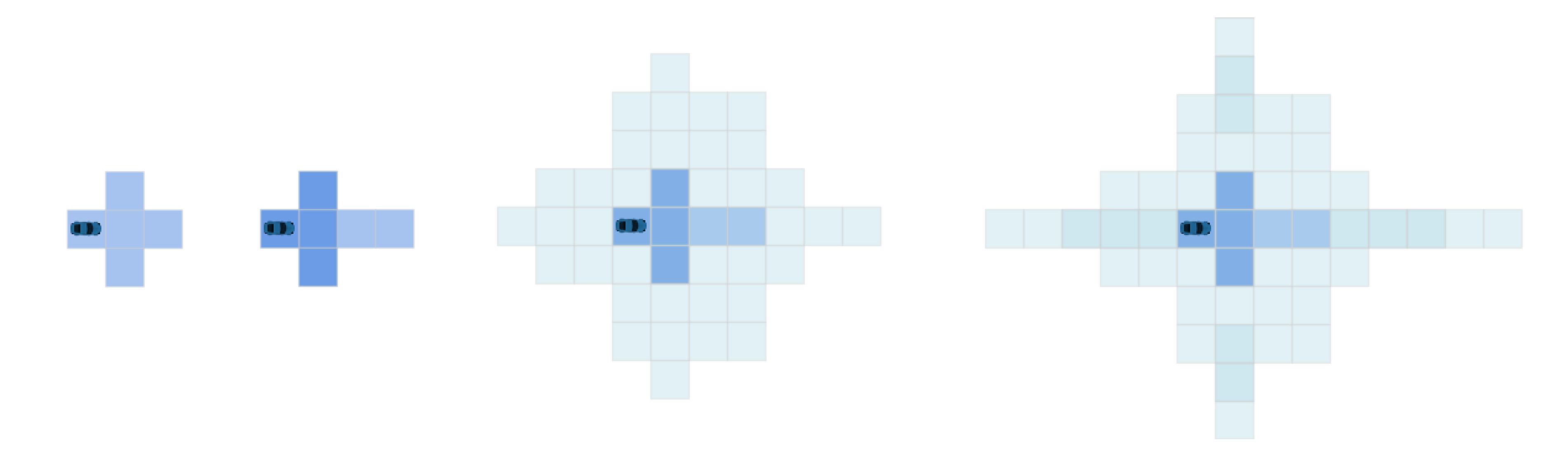}
\caption{Bubble if all $\text{Ag}\in\mathfrak{A}$ have the Agent Dynamics specified in Section \ref{agent_attributes}. Construction of this set defined in the Appendix.}
\label{fig_build_bubbles}
\end{figure}

We begin by considering the ego agent whose bubble we are defining. In particular, let us again consider an agent $\text{Ag}$ at state $s_0 \in S_{Ag}$. The corresponding grid point set $\mathcal{G}_{\text{Ag}}^{\mathcal{R}}(s_0)$ is shown in the left-most figure in Fig. \ref{fig_build_bubbles}. The grid points an agent occupies when executing its backup plan from a state in the agent's forward reachable set $\mathcal{R}_{\text{Ag}}(s_0)$ is given by:
\begin{equation*}
\mathcal{G}^{\R,BP}_{\text{Ag}}(s_0) \triangleq \bigcup_{s \in \R_{\text{Ag}}(s_0)} BP_{\text{Ag}}(s)
\end{equation*}
These grid points are shown in the second from the left sub-figure in Fig. \ref{fig_build_bubbles}. The set-valued map $$\mathcal{Z}_{\text{Ag}}(s_0) \triangleq \mathcal{G}^{\mathcal{R}}_{\text{Ag}}(s_0) \cup \mathcal{G}^{\R,BP}_{\text{Ag}}(s_0).$$ 
represents all the grid points an agent can possibly reach in the next state or in the following time step were it to execute its backup plan. 
Let $\text{Ag}' \in \mathfrak{A}$ and $\text{Ag}' \neq \text{Ag}$. The set:
\begin{equation*}
\mathcal{S}_{\text{Ag}'}^{\mathcal{R}}(\text{Ag}, s_0) \triangleq \bigcup_{n \in \mathcal{Z}_{\text{Ag}}(s_0)}
\mathcal{G}^{\R^{-1}}_{\text{Ag}'}(n)
\end{equation*} 
defines the set of all states in which another agent $\text{Ag}'$ can reach any grid point in the other agents' forward reachable grid points $\mathcal{Z}_{\text{Ag}}(s_0)$. Let us define the grid point projection of these states as 
$$\mathcal{G}_{\text{Ag}'}^{\mathcal{R}}(\text{Ag}, s_0) \triangleq \{\mathcal{G}_{\text{Ag}'}(s) \mid s \in \mathcal{S}_{\text{Ag}'}^{\mathcal{R}}(\text{Ag}, s_0) \}.$$
These grid points are defined in the third from the left subfigure in Fig. \ref{fig_build_bubbles}. 

The bubble also needs to include any state where an agent $\text{Ag}'$ where the agent has so much momentum it cannot stop fast enough to avoid collision with the agent $\text{Ag}$. To define the set of states from which this might occur, let us define the set:  
\begin{equation*}
\mathcal{S}^{BP}_{Ag'}(Ag, s_0) = \{s \in S_{Ag'} \mid BP_{Ag'}(s) \cap \mathcal{Z}_{Ag}(s_0) \neq \emptyset\}.
\end{equation*}
If another agent $\text{Ag}'$ occupies a state in this set, then execution of that agent's backup plan will cause it to intersect with the set of grid points that are in agents set $\mathcal{Z}_{\text{Ag}}(s_0)$. 
Let 
\begin{equation*}
\mathcal{S}^{\R, BP}_{\text{Ag}'}(\text{Ag}, s_0) =  \bigcup_{s \in \mathcal{S}^{BP}_{\text{Ag}'}(\text{Ag})} \R_{\text{Ag}'}^{-1}(s).
\end{equation*}
This is the set of all states backward reachable to the states in $\mathcal{S}^{BP}_{\text{Ag}'}(\text{Ag}, s_0)$. If an agent $\text{Ag}'$ occupies any of these states, it will end up in a state where its backup plan will intersect with agent $\text{Ag}$'s potential grid points that are defined in $\mathcal{Z}_{\text{Ag}}$. We project this set of states to a set of grid points as 
$$\mathcal{G}^{\R, BP}_{\text{Ag}'}(\text{Ag}, s_0) = \{\mathcal{G}_{Ag'}(s) \mid s \in \mathcal{S}^{BP}_{Ag'}(Ag, s_0)\}.$$
Note, this set of grid points is shown in the right-most subfigure in Fig. \ref{fig_build_bubbles}. The bubble is then defined as the union of all the sets of grid points specified above.
\begin{definition}[Bubble]
Let us consider an agent $\text{Ag}$ with state $s_0 \in S_{\text{Ag}}$ and agent $\text{Ag}'$ be another agent. Then the bubble of $\text{Ag}$ with respect to agents of the same type as $\text{Ag}'$ is given by $$\mathcal{B}_{\text{Ag}/\text{Ag}'}(s_0) \triangleq \mathcal{Z}_{\text{Ag}}(s_0) \cup \mathcal{G}_{\text{Ag}'}^{\mathcal{R}}(\text{Ag}, s_0) \cup \mathcal{G}^{\R, BP}_{\text{Ag}'}(\text{Ag}, s_0).$$
\end{definition}
Note that under almost all circumstances, we should have
$$
\mathcal{Z}_{\text{Ag}}(s_0) \subseteq \mathcal{G}^{\mathcal{R}}_{\text{Ag}'}(Ag, s_0) \subseteq \mathcal{G}^{\R, BP}_{\text{Ag}'}(\text{Ag}, s_0)
$$ so $\mathcal{B}_{\text{Ag}}(s_0)$ is simply equal to $\mathcal{G}^{\R, BP}_{\text{Ag}'}(\text{Ag}, s_0)$. This holds true for the abstract dynamics we consider in this paper.
This means the bubble contains any grid points in which another agent $\text{Ag}'$ occupying those grid points can interfere (via its own forward reachable states or the backup plan it would use in any of its forward reachable states) with at least one of agent $\text{Ag}$'s  next possible actions and the backup plan it would use if it were to take any one of those next actions. 

\subsection{Global Precedence Consistency}
\begin{lemma}
\label{global_precedence}
If all agents assign precedence according to the local precedence assignment rules to agents in their respective bubbles, then the precedence relations will induce a polyforest on $\mathfrak{A}/\sim$, where $S/\sim$ defines the quotient set of a set $S$.
\end{lemma}
\begin{proof}
Suppose there is a cycle $C$ in $\mathfrak{A}/\sim$. For each of the equivalent classes in $C$ ($C$ must have at least $2$ to be a cycle), choose a representative from $\mathfrak{A}$ to form a set $R_C$. Let $\text{Ag} \in R_C$ be one of these representatives. Applying the second local precedence assignment rule inductively, we can see that all agents in $R_C$ must be from $\text{Ag}$'s bundle. By the first local precedence assignment rule, any $C$ edge must be from an agent with lower projected value to one with a higher projected value in this bundle. Since these values are totally ordered (being integers), they must be the same. This implies that $C$ only has one equivalence class, a contradiction.
\end{proof}
The acyclicity of the polyforest structure implies the consistency of local agent precedence assignments. Note, the local precedence assignment algorithm establishes the order in which agents are taking turns. 

\subsection{Oracle Definitions}
\begin{enumerate}
    \item \textbf{$O_{\text{Ag},t,\text{unprotected left-turn safety}}(s, a, u)$} returns \texttt{T} when the action $a$ from the state $s$ will result in the complete execution of a safe, unprotected left-turn (invariant to agent precedence). Note, an unprotected left turn spans over multiple time-steps. The oracle will return \texttt{T} if $\text{Ag}$ has been waiting to take left-turn (while traffic light is green), traffic light turns red, and no agents in oncoming lanes. 
    \item \textbf{$O_{\text{static safety}}(s, a, u)$} returns \texttt{T} when the action $a$ from state $s$ will not cause the agent to collide with a static obstacle or end up in a state where the agent's safety backup plan $a_{bp}$ with respect to the static obstacle is no longer safe. 
    \item \textbf{$O_{\text{traffic light law}}(s, a, u)$} returns \texttt{T} if the action $a$ from the state $s$ satisfies the traffic light laws (not crossing into intersection when red. It also requires that $\text{Ag}$ be able to take $a_{bp}$ from $s' = \tau_{\text{Ag}}(s,a)$ and not violate the traffic-light law.
    \item \textbf{$O_{\text{traffic orientation law}}(s, a, u)$}  returns \texttt{T} if the action $a$ from the state $s$ follows the legal road orientation.
    \item \textbf{$O_{\text{traffic intersection clearance law}}(s, a, u)$} returns \texttt{T} if the action causes the agent to enter the intersection and not leave it when the traffic light turns red. Returns \texttt{T} if the action causes the agent to end in a state where its backup plan action will cause the agent to enter the intersection and not be able to leave it when the traffic light turns red. 
    \item \textbf{$O_{\text{traffic intersection lane change law}}(s, a, u)$} returns \texttt{T} if the action is such that \\ $\gamma_{Ag} = \{\texttt{left-lane change}, \texttt{right-lane change} \}$ and the agent either begins in an intersection or ends up in the intersection after taking the action.
    \item \textbf{$O_{\text{maintains progress}}(s, a, u)$} returns \texttt{T} if the action $a$ from the state $s$ stays the same distance to its goal.
\end{enumerate}

\subsection{Action Selection Strategy}
\begin{figure}[H]
\centering
\includegraphics[scale=0.20]{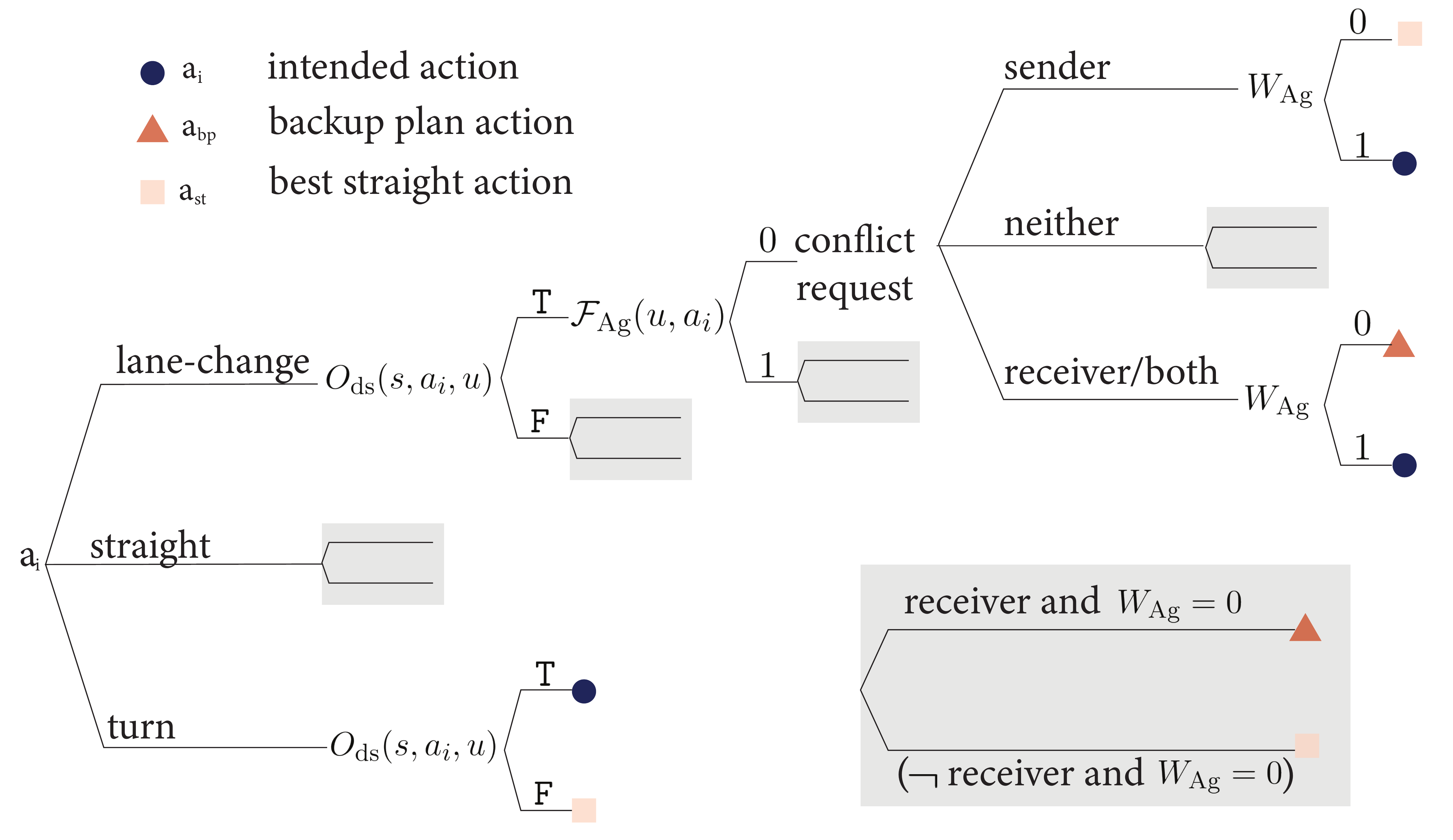}
\caption{Agent action selection strategy.}
\label{fig_action_selection_strategy}
\end{figure}

\subsection{Safety Lemmas}
In the following lemma, we show that an agent cannot send (or receive) a conflict request to (from) an agent outside its bubble. 
\begin{lemma}
\label{conflict_bubble_lemma}
Let us consider agent $\text{Ag}$ with state $s$ and agent $\text{Ag}'$ at state $s'$. $\text{Ag} \texttt{ send } \text{Ag}' \Rightarrow \text{Ag} \in \mathcal{B}_{\text{Ag}'}(s')$. 
\end{lemma}
\begin{proof}
If $A \texttt{ send } B$ this means that all of the conditions specified in Section \ref{agent_action_conflict}, particularly that $(A, a_i) \dagger (B, a_i')$. This condition is only valid if $\text{proj}_G s \in \mathcal{G}_{F, B}(B,A)$ or
$\text{proj}_G s \in \mathcal{G}_{F, BP}(B,A)$. Membership of Agent A's state in either of these sets implies $A \in \mathcal{B}(B)$.
\end{proof}

The following lemma follows from the lemma above. 

\begin{lemma}
\label{at_most_conflict_cluster}
At most one agent will win in each agent's conflict cluster. 
\end{lemma}
\begin{proof}
W.l.o.g. let us consider an agent $\text{Ag}$ and its respective conflict cluster $\mathcal{C}(\text{Ag})$. It follows from Lemma \ref{conflict_bubble_lemma} that $\forall \text{Ag}'$, s.t. $ \text{Ag} \texttt{ send } \text{Ag}' \text{Ag}' \in \mathcal{B}_{Ag}(s)$ and $\text{Ag} \in \mathcal{B}_{Ag'}(s')$. It also follows that $\forall \text{Ag}' \text{ s.t. }, \text{Ag} \texttt{ send } \text{Ag}', \text{Ag} \in \mathcal{B}_{\text{Ag}'}(s')$ and $\text{Ag}' \in \mathcal{B}_{\text{Ag}}(s)$. This means an agent has access to all token counts and IDs of all agents in its conflict cluster, and all agents in its conflict cluster have access to the agent's token count and ID. The conflict resolution implies that all agent edges are incident to the winning agent, where edges point to the agent they cede to. This implies that at most one agent can be the winner of each cluster. Less than one winner (per conflict cluster) will occur when an agent that is in the intersection of more than one conflict cluster wins. 
\end{proof}

The following lemma states that if all $\text{Ag} \in \mathfrak{A}$ are  following the Agent Protocol, an agent $\text{Ag}$ will not take an action that will cause it to 1) collide with or 2) violate the safety backup plan of another agent outside its bubble $\mathcal{B}_{\text{Ag}}(s)$.
\begin{lemma}
\label{lemma_safety_agents_outside_bubble}
If $\text{Ag}$ is following the Agent Protocol, and $S_{\text{Ag}, bp}(u) = \texttt{T}$, $\text{Ag}$ will only choose an action $a \in Act_{\text{Ag}}$ for which the following two conditions hold: 1) $\mathcal{G}_{\text{Ag}}(s, a) \cap (\cup_{\text{Ag}'\in S}\mathcal{G}_{\text{Ag}'}(s', a')) = \emptyset$ and 2) $\forall \text{Ag}'\in S$, $\lnot ((\text{Ag}, a)\bot \text{Ag}')$, where the set $S \triangleq \{\text{Ag}'|\text{Ag}' \notin \mathcal{B}_{\text{Ag}}(s) \land ((\text{Ag}' \sim \text{Ag}) \lor (\text{Ag}' \prec \text{Ag}) \lor (\text{Ag} \prec \text{Ag'})) \}$.
\end{lemma}
\begin{proof}
This follows from the definition of the agent bubble, whose construction is defined in \ref{bubble_construction}.
\end{proof}

The following lemma states that an agent $\text{Ag}$ following the Agent Protocol will not take an action for which it violates the safety of its own backup plan.
\begin{lemma}
\label{lemma_safety_itself}
If $\text{Ag}$ is following the Agent Protocol, and $S_{\text{Ag}, bp}(u) = \texttt{T}$, $\text{Ag}$ will only choose an action $a \in Act_{\text{Ag}}$ for which the following condition holds: $\forall \text{Ag}'\in S$, $\lnot ((\text{Ag}, a)\bot Ag')$, where $S = \{\text{Ag}\}$.
\end{lemma}
\begin{proof}
We prove this by using specific definition of elements in the Agent Protocol.
\begin{enumerate}
    \item Let us first show that any action $a\in Act_{\text{Ag}}$ that $\text{Ag}$ takes will satisfy the oracles in the top two tiers (safety and traffic rules) of $\text{Ag}$'s profile defined in Section. \ref{section_profiles}.\\
    \label{safety_top_two_tiers}
    \begin{enumerate}
        \item According to the Action Selection Strategy defined in Section \ref{section_action_selection_strategy}, $\text{Ag}$ will choose one of three actions: the agent's intended action $a_i$, the best straight action $a_{st}$, or its backup plan action $a_{bp}$.
        \item Let us consider the actions $a_i$ and $a_{st}$. 
        \label{safety:intended_and_straight}
        \begin{enumerate}
             \item Both $a_i$ and $a_{st}$ are selected via the Agent Profile and consistent-function evaluator defined in Section \ref{section_profiles}.
            \label{safety_consistent_function_evaluator}
            \item Since $S_{\text{Ag}, bp}(u) = \texttt{T}$, the agent will have at least one action ($a_{bp}$) for which the top two tiers of specifications are satisfied. 
            \label{safety:assumption_have_action}
            \item By definition of the Agent Profile and the consistent evaluator function, if $S_{\text{Ag}, bp}(u) = \texttt{T}$, the safety backup plan action $a_{bp}$ will always be chosen over an action where any of the specifications in the top two tiers of the profile are not satisfied. \label{safety:agent_profile_choice}
            \item By \ref{safety:assumption_have_action} and \ref{safety:agent_profile_choice}, $\text{Ag}$ will have $a \in  Act_{Ag}$ and will choose an action for which the top two tiers of the Agent Profile are satisfied and thus $a_{i}$ and $a_{st}$ are actions where all oracles in the top two tiers of the profile are satisfied.
        \end{enumerate}
    \item Let us consider the action $a_{bp}$. 
    \begin{enumerate}
            \item This follows from the assumption that $S_{\text{Ag}, bp}(u) = \texttt{T}$ and the definition of $S_{\text{Ag}, bp}(u)$.
            \label{safety:max_yielding_eval_2}
    \end{enumerate}
    \end{enumerate}
    \item If the oracles in the top two tiers are satisfied by an action $a$, by the definition of the oracles in Section \ref{section_profiles}, this implies that the action $a$ will take $\text{Ag}$ to a state $s'$ and the system will be in a new global state $u'$ where $S_{Ag, bp}(u')= \texttt{T}$. 
    \item $S_{Ag, bp}(u')= T$ means $\text{Ag}$ will end up in a state where $a_{bp}$ will be an action that satisfies traffic rules, avoids inevitable collision with static obstacles, and $\lnot((Ag,a_i) \bot \text{Ag})$. 
    \label{safety_top_two_tiers_end}
\end{enumerate}
\end{proof}

The following lemma states that if all $\text{Ag} \in \mathfrak{A}$ are  following the Agent Protocol, any agent $\text{Ag}$ will not take an action for which it collides with or violates the safety backup plan of any agent with higher precedence.
\begin{lemma}
\label{lemma_safety_higher}
If $\text{Ag}$ is following the Agent Protocol, and $S_{\text{Ag}, bp}(u) = \texttt{T}$, $\text{Ag}$ will only choose an action $a \in Act_{\text{Ag}}$ for which the following two conditions hold: 1) $\mathcal{G}_{Ag}(s, a) \cap (\cup_{Ag'\in S}\mathcal{G}_{\text{Ag}'}(s', a')) = \emptyset$ and 2) $\forall \text{Ag}'\in S$, $\lnot ((Ag, a)\bot Ag')$, where the set $S \triangleq \{\text{Ag}'|\text{Ag} \prec \text{Ag}'\}$, i.e. agents with higher precedence than $\text{Ag}$.
\end{lemma}

\begin{proof}
We prove this by using arguments based on the definition of precedence, the Agent Protocol, and Agent Dynamics. 
    \begin{enumerate}
        \item Let us first consider all $\text{Ag}'$ where $\text{Ag} \prec \text{Ag}'$ and $Ag' \notin \mathcal{B}_{\text{Ag}}(s)$.
        \begin{enumerate}
            \item Proof by Lemma \ref{lemma_safety_agents_outside_bubble}. 
        \end{enumerate}
        \item Now, let us consider all $\text{Ag}'$ where $\text{Ag} \prec \text{Ag}'$ and $\text{Ag}' \in \mathcal{B}_{\text{Ag}}(s)$.
        \item According to Lemma \ref{lemma_safety_itself}, $\text{Ag}$ will only take an action that satisfies all oracles in the top two tiers, including \\ $O_{\text{dynamic safety}}(s, a, u)$. 
        \item Since $a$ is such that $O_{\text{dynamic safety}}(s, a, u) = \texttt{T}$, by definition of the oracle, $Ag$ will not cause collision with any $Ag' \in \mathcal{B}_{\text{Ag}}(s)$.
        \item For any $\text{Ag} \prec \text{Ag}'$, where $\text{Ag}'$ has higher precedence than $\text{Ag}$, then $\text{proj}_{\text{long}}(\text{Ag}) < \text{proj}_{\text{long}} (\text{Ag}')$, i.e. $\text{Ag}'$ is longitudinally ahead of $\text{Ag}$. 
        \item In order for $(\text{Ag}, a) \bot \text{Ag}'$, the action $a$ would have to be such that $s_f =  \tau_{Ag}(s,a)$, and $ La(s_f) = La(s')$ and $\text{proj}_{\text{long}}(\text{Ag}) > \text{proj}_{\text{long}} (\text{Ag}')$, where $\text{Ag}$ is directly in front of $\text{Ag}'$.  
        \item Because of the agent dynamics defined in Section \ref{agent_attributes}, any $a$ such that $(\text{Ag}, a) \bot \text{Ag}'$ will require $\mathcal{G}(\text{Ag},a) \cap \mathcal{G}(\text{Ag}') \neq \emptyset$. 
        \item Thus, any such action $a$ will not satisfy the oracle \\
        $O_{\text{dynamic safety}}(s, a, u)$.
        \item Since $S_{\text{Ag}, bp}(u) = \texttt{T}$, by Assumption 6 in Section \ref{all_assumptions}, the agent will have at least one action $a_{bp}$ for which \\
        $O_{\text{dynamic safety}}(s, a, u) = \texttt{T}$.
        \item Since the agent will only choose an action for which \\
        $O_{\text{dynamic safety}}(s,a,u) = \texttt{T}$ and it always has at least one action $a_{bp}$ that satisfies the oracle, the agent will always choose an action for which $O_{\text{dynamic safety}}(s,a,u) = \texttt{T}$ and thus will take an action such that $\lnot ((\text{Ag}, a) \bot \text{Ag}')$.
    \end{enumerate}
\end{proof}

The following lemma states that if all $\text{Ag} \in \mathfrak{A}$ are  following the Agent Protocol, any agent $\text{Ag}$ will not take an action for which it collides with or violates the safety backup plan of any agent with lower precedence.

\begin{lemma}
\label{lemma_safety_lower}
If $\text{Ag}$ is following the Agent Protocol, and $S_{\text{Ag}, bp}(u) = \texttt{T}$, $\text{Ag}$ will only choose an action $a \in Act_{\text{Ag}}$ for which the following two conditions hold: 1) $\mathcal{G}_{Ag}(s, a) \cap (\cup_{Ag'\in S}\mathcal{G}_{\text{Ag}'}(s', a')) = \emptyset$ and 2) $\forall Ag'\in S$, $\lnot ((Ag, a)\bot Ag')$, where the set $S \triangleq \{\text{Ag}' | \text{Ag}' \prec \text{Ag} \}$, i.e. agents with lower precedence than $\text{Ag}$.
\end{lemma}
\begin{proof}
We prove this by using arguments based on the definition of precedence, the Agent Protocol, and Agent Dynamics. 
\begin{enumerate}
    \item Let us first consider all $\text{Ag}'$ where $\text{Ag} \prec \text{Ag}'$ and $\text{Ag}' \notin \mathcal{B}_{\text{Ag}}(s)$.
    \begin{enumerate}
        \item Proof by Lemma \ref{lemma_safety_agents_outside_bubble}.
    \end{enumerate}
    \item Now, let us consider all $\text{Ag}'$ where $\text{Ag} \prec \text{Ag}'$ and $\text{Ag}' \in \mathcal{B}_{\text{Ag}}(s)$.
    \item According to \ref{safety_top_two_tiers_end}, $\text{Ag}$ will only take an action that satisfies all oracles in the top two tiers, including $O_{\text{dynamic safety}}(s,a,u)$. 
    \item Since $a$ is such that $O_{\text{dynamic safety}}(s, a, u) = \texttt{T}$, by definition of the oracle, $Ag$ will not cause collision with any $Ag' \in \mathcal{B}_{\text{Ag}}(s)$.
    \item According to the Action Selection Strategy defined in Section \ref{section_action_selection_strategy}, $\text{Ag}$ will choose one of three actions: the agent's intended action $a_i$, the best straight action $a_{st}$, or its backup plan action $a_{bp}$.
    \item Let us consider the backup plan action $a_{bp}$.
    \label{safety_lower_backup_plan}
    \begin{enumerate}
        \item By violation of safety backup plan, $((\text{Ag}, a_{bp}) \bot \text{Ag}')$ only if $La(\text{Ag}) = La(\text{Ag}')$.
        \item W.l.o.g., let us consider $\text{Ag}'$ that is directly behind $\text{Ag}$.
        \item Since $S_{\text{Ag}', bp}(s, u) = \texttt{T}$, by Assumption 6 in Section \ref{all_assumptions}, $O_{\text{dynamic safety}}(s, a_{bp},u) = \texttt{T}$, meaning $\text{Ag}'$ will be far enough behind $\text{Ag}$ so that if $\text{Ag}$ executes its backup plan action $a_{bp}$, $\text{Ag}'$ can safely execute its own backup plan action. 
        \item Thus, by Definition \ref{safety_backup_plan_violation_action}, $\lnot((\text{Ag}, a_{bp}) \bot \text{Ag}')$.
    \end{enumerate}
    \item Let us consider the best straight action $a_{st}$.
    \begin{enumerate}
        \item This follows from the arguments made in \ref{safety_lower_backup_plan}, since $a_{st}$ is a less severe action than $a_{bp}$.
    \end{enumerate}
    \item Let us consider the intended action $a_i$.
    \begin{enumerate}
        \item Let us consider when $\gamma_{\text{Ag}} = \{\texttt{straight}\}$. 
        \begin{enumerate}
            \item This follows from \ref{safety_lower_backup_plan}.
        \end{enumerate}
        \item Let us consider when $\gamma_{\text{Ag}} \in \{\texttt{right-turn}, \texttt{left-turn} \}$.
        \begin{enumerate}
            \item If $\text{Ag}$ takes such an action, $\text{Ag}$ will end up in a state where $Bu(\text{Ag}') \neq Bu(\text{Ag})$ and from Definition \ref{safety_backup_plan_violation_action}, agents in different bundles cannot violate each others' backup plans. 
        \end{enumerate}
        \item Let us consider when $\gamma_{\text{Ag}} \in \{\texttt{right-lane change}$  \\$\texttt{left-lane change} \}$.
        \begin{enumerate}
            \item $(\text{Ag}, a_i) \bot \text{Ag}'$ when $a_i$ is a lane change and the agents $\text{Ag}$ and $\text{Ag}'$ are at a state such that $s_f = \tau(s, a_i)$ and $s_f' = \tau(s', a_{bp})$, respectively, where $d(s_f, s_f') < gap_{req}$, where $d(s_f, s_f')$ is the $l_2$ distance between $s_f$ and $s_f'$.
            \item When this condition holds, the agent's max-yielding-not-enough flag $\mathcal{F}_{\text{Ag}}(u, a_i)$ defined in Section \ref{definition_max_yielding_flag_not_enough} will be set.
            \item According to the action-selection strategy, $\text{Ag}$ will only take $a_i$ when  $\mathcal{F}_{\text{Ag}}(u, a_i) = \texttt{F}$. 
            \item Thus, $\text{Ag}$ will only take $a_i$ when $\lnot((Ag, a_{i}) \bot Ag')$.
        \end{enumerate}
    \end{enumerate}
\end{enumerate}
\end{proof}

The following lemma states that if all $\text{Ag} \in \mathfrak{A}$ are  following the Agent Protocol, any agent $\text{Ag}$ will not take an action for which it collides with or violates the safety backup plan of any agent with equal precedence.

\begin{lemma}
\label{lemma_safety_equal}
If $\text{Ag}$ is following the Agent Protocol, and $S_{\text{Ag}, bp}(u) = \texttt{T}$, $\text{Ag}$ will only choose an action $a \in Act_{\text{Ag}}$ for which the following two conditions hold: 1) $\mathcal{G}_{\text{Ag}}(s, a) \cap (\cup_{\text{Ag}'\in S}\mathcal{G}_{\text{Ag}'}(s', a')) = \emptyset$ and 2) $\forall \text{Ag}'\in S$, $\lnot ((\text{Ag}, a)\bot \text{Ag}')$, where the set $S \triangleq \{\text{Ag}' | \text{Ag}' \sim \text{Ag} \}$, i.e. agents with equivalent precedence as the agent.
\end{lemma}

\begin{proof}
We prove this by using arguments based on the definition of precedence, Agent Dynamics, and the Agent Protocol. 
\begin{enumerate}
    \item Let us first consider all $\text{Ag}'$ where $\text{Ag} \prec \text{Ag}'$ and $\text{Ag}' \notin \mathcal{B}_{\text{Ag}}(s)$.
    \begin{enumerate}
        \item Proof by Lemma \ref{lemma_safety_agents_outside_bubble}.
    \end{enumerate}
    \item Now, let us consider all $\text{Ag}'$ where $\text{Ag} \prec \text{Ag}'$ and $\text{Ag}' \in \mathcal{B}_{\text{Ag}}(s)$.
    \item Let us first consider the agent itself, since an agent has equivalent precedence to itself. 
    \begin{enumerate}
        \item This is true by Lemma \ref{lemma_safety_itself}.
    \end{enumerate}
    \item This can be proven for any other agents of equivalent precedence that is not the agent itself as follows.
    \item Agents with equal precedence take actions simultaneously so $O_{\text{dynamic safety}}(s,a,u)$ does not guarantee no collision.
    \item According to the Action Selection Strategy defined in Section \ref{section_action_selection_strategy}, $\text{Ag}$ will choose one of three actions: the agent's intended action $a_i$, the best straight action $a_{st}$, or its backup plan action $a_{bp}$.
    \item By definition of precedence assignment, any $\text{Ag}'$ for which $\text{Ag}' \sim \text{Ag}$ will be such that $La(\text{Ag}) \neq La(\text{Ag}')$.
    \label{equal_precedence_separate_lanes}
    \item Let us show if $\text{Ag}$ selects $a_{bp}$, it will 1) not collide with any $\text{Ag}'\in S$ and 2) $\lnot ((\text{Ag}, a_{bp}) \bot \text{Ag}')$.
    \label{equal_backup_plan_action}
    \begin{enumerate}
        \item W.l.o.g., let us consider $\text{Ag}'$ where $\text{Ag}' \sim Ag$.
        \item The flag $\mathcal{F}_{\text{Ag'}}(u, a_i) = \texttt{T}$ if $\text{Ag}'s$ intended action $a_i$ causes collision with $\text{Ag}$ or $(Ag', a_{i}) \bot Ag$, i.e. it collides with or violates the safety of $\text{Ag}$'s backup plan action.
        \item By the action-selection-strategy, $\text{Ag}'$ will not take the action $a_i$ when $\mathcal{F}_{\text{Ag'}}(u,a_i) = \texttt{T}$, so this guarantees $\text{Ag}$ will not collide with $\text{Ag}'$ when $\text{Ag}$ takes $a_{bp}$.
        \item By the Agent Dynamics, $\text{Ag}$'s backup plan action cannot cause $\text{Ag}$ to end up in a position where it can violate $\text{Ag}'$'s backup plan without colliding with it--for which $\text{Ag}'$'s flag $\mathcal{F}_{\text{Ag}}(u, a_i)$ would be set. 
    \end{enumerate}
    \item Let us show that $\text{Ag}$ will only choose an $a_{st}$ if it will 1) not collide with $\text{Ag}'\in S$ and 2) $\lnot ((\text{Ag}, a_{st}) \bot \text{Ag}')$.
    \label{safety_equal_straight}
    \begin{enumerate}
        \item When $a_{st} = a_{bp}$, then the arguments in \ref{equal_backup_plan_action} hold.
        \item $\text{Ag}$ selects an $a_{st}$ that is not $a_{bp}$ only when 1) its conflict cluster is empty (i.e. $C_{\text{Ag}} = \emptyset$) or 2) when it has received a conflict request from another agent and it has won its conflict cluster resolution (i.e. $W_{\text{Ag}} = \texttt{T}$). 
        \item If $C_{\text{Ag}} = \emptyset$, by definition of how conflict clusters are defined in Section \ref{conflict_cluster_definition}, the agent's action $a_{st}$ will not cause $\text{Ag}$ to collide with any $\text{Ag}' \in S$, and $\forall \text{Ag}' \in S, \lnot((\text{Ag}, a_{st}) \bot \text{Ag}')$.
        \item In the case $\text{Ag}$ has received a conflict request and has won $W_{\text{Ag}}$, by Lemma \ref{conflict_bubble_lemma}, if $W_{\text{Ag}} = \texttt{T}$, it will be the only agent in its conflict cluster that has won.
        \label{straight_action_safety_2}
        \item By definition of the conflict cluster, any $\text{\text{Ag}}' \in C_{\text{Ag}}$ where $\text{Ag} \sim \text{Ag}'$ will take a straight action. 
        \item Since agents of equivalent precedence are initially in separate lanes by \ref{equal_precedence_separate_lanes} and any $\text{Ag}' \in S$ will take a straight action, then $La(s_{\text{Ag}, t+1}) \neq La(s_{\text{Ag'}, t+1})$ when $\text{Ag}$ takes $a_{st}$.
        \item Thus, by definition of agent dynamics and Definition \ref{safety_backup_plan_violation_action}, the action will not cause $\text{Ag}$ to collide with any $\text{Ag}' \in S$, and $\forall \text{Ag}' \in S, \lnot((\text{Ag}, a_{st}) \bot \text{Ag}')$.
    \end{enumerate}
    \item Let us show that $\text{Ag}$ will only choose an $a_{i}$ if it will 1) not collide with any $\text{Ag}' \in S$ and 2) $\lnot ((\text{Ag}, a_{i}) \bot \text{Ag}')$.
    \begin{enumerate}
    \item Let us consider when $\gamma_{\text{Ag}} = \texttt{straight}$ for $a_i$.
    \begin{enumerate}
        \item This follows from the same arguments presented in \ref{safety_equal_straight}.
    \end{enumerate}
    \item Let us consider when $\gamma_{\text{Ag}} \in \{\texttt{right-turn, left-turn}\}$ for $a_i$. 
    \begin{enumerate}
        \item This follows from the fact that all other agents are following the Agent Protocol and will not take a lane-change action in the intersection, and because of the definition of the Agent Dynamics and Road Network.
    \end{enumerate}
    \item Let us consider when  $\gamma_{\text{Ag}} \in \{ \texttt{right-lane change}$, \\$\texttt{left-lane change} \}$.
    \begin{enumerate}
        \item $\text{Ag}$ will only take its intended action $a_i$ if the flag \\ $\mathcal{F}_{\text{Ag}}(u, a_i) = \texttt{F}$, and in the case that it is part of a conflict cluster, it is the winner of the conflict cluster resolution, i.e. $\mathcal{W}_{\text{Ag}} = \texttt{T}$.
        \item By definition of $\mathcal{F}_{\text{Ag}}(u, a_i)$, the agent will not take $a_i$ when $a_i$ causes $\text{Ag}$ to collide with any agent $\text{Ag}' \in S$ or when it causes $Ag$ to violate the safety of the back up plan of another agent $Ag'$, i.e. $ \exists Ag'$ s.t. $(\text{Ag}, a_i) \bot \text{Ag}'$.
        \item In the case the agent has received a conflict request and has won $\mathcal{W}_{\text{Ag}}$, by Lemma \ref{conflict_bubble_lemma}, if $\mathcal{W}_{\text{Ag}} = \texttt{T}$, it will be the only agent in its conflict cluster that has won.
        \item By definition of the conflict cluster, any $\text{Ag}' \in C_{\text{Ag}}$ where $\text{Ag} \sim \text{Ag}'$ will take its backup plan action $a_{bp}$, and thus $s_f = \tau(s, a_{st})$, and $s_f' = \tau(s, a_{bp})$, where \\
        $ d(s_f, s_f') \geq gap_{\text{req}} $. 
        \item Thus, $a_i$ will only be selected when $a_i$ does not cause $\text{Ag}$ to collide with any $\text{Ag}' \in S$ and \\ $\forall \text{Ag}' \in S, \lnot((\text{Ag}, a_i) \bot \text{Ag}')$. 
    \end{enumerate}
    \end{enumerate}
\end{enumerate}
\end{proof}

The following lemma states that if all $\text{Ag} \in \mathfrak{A}$ are  following the Agent Protocol, any agent $\text{Ag}$ will not take an action for which it collides with or violates the safety backup plan of any agent with incomparable precedence to it.
\begin{lemma}
\label{lemma_safety_incomparable}
If $\text{Ag}$ is following the Agent Protocol, and $S_{\text{Ag}, bp}(u) = \texttt{T}$, $\text{Ag}$ will only choose an action $a \in Act_{\text{Ag}}$ for which the following two conditions hold: 1) $\mathcal{G}_{\text{Ag}}(s, a) \cap (\cup_{\text{Ag}'\in S}\mathcal{G}_{\text{Ag}'}(s', a')) = \emptyset$ and 2) $\forall \text{Ag}'\in S$, $\lnot ((\text{Ag}, a)\bot \text{Ag}')$, where the set $S \triangleq \{\text{Ag}' | \text{Ag}' \not \sim \text{Ag} \}$, i.e. agents with precedence incomparable to the agent.
\end{lemma}
\begin{proof}
We prove this by using arguments based on the definition of precedence, Agent Dynamics, and the Agent Protocol. 
\begin{enumerate}
\item Let us show when $\text{Ag}$ chooses $a_{bp}$, it will 1) not collide with any $\text{Ag}' \in S$ and 2) $\lnot ((\text{Ag}, a_{bp}) \bot \text{Ag}')$.
    \label{safety_incomparable_backup}
    \begin{enumerate}
        \item Since $S_{\text{Ag}, bp}(u) = \texttt{T}$, the agent will have at least one action ($a_{bp}$) for which the top two tiers of specifications are satisfied.
        \label{assumption_backup_plan}
        \item By \ref{assumption_backup_plan}, the action $a_{bp}$ will only take $\text{Ag}$ into the  intersection if traffic light is green. 
        \item By Assumption 4, all traffic lights are coordinated so if agents respect traffic light rules, they will not collide.
        \item By the assumption that all other $\text{Ag}' \in \mathfrak{G}$ are obeying the same protocol, each agent will only take actions that satisfy the top two tiers of their profile.
        \item Any $\text{Ag}'$ in a perpendicular bundle will not enter the intersection since they have a red light. 
        \item Thus, $\text{Ag}$ cannot collide or violate the backup plan of agents in perpendicular bundles.
        \item Any $\text{Ag}'$ in an oncoming traffic bundle must only take an unprotected left-turn when it satisfies \\ $O_{\text{unprotected left-turn}}(s,a,u)$.
        \item Thus $\text{Ag}$ will not collide or violate the backup plan of agents in bundles of oncoming traffic.
        \end{enumerate}
    \item Let us show that when $\text{Ag}$ chooses $a_{st}$, it will 1) not collide with any $\text{Ag}' \in S$ and 2) $\lnot ((\text{Ag}, a_{st}) \bot \text{Ag}')$.
    \label{safety_incomparable_straight}
    \begin{enumerate}
        \item Since $a_{st}$ is chosen according to the Agent Profile, it will only be a straight action that is not $a_{bp}$ as long as it satisfies the top-two tiers of the profile and more. 
        \item Thus, $a_{st}$ will only take $\text{Ag}$ into intersection if traffic light is green.  
        \item By the same arguments in \ref{safety_incomparable_backup}, this holds.
    \end{enumerate}
    \item Let us show that when $\text{Ag}$ chooses $a_{i}$, it will 1) not collide with any $\text{Ag}' \in S$ and 2) $\lnot ((\text{Ag}, a_{i}) \bot \text{Ag}')$.
    \begin{enumerate}
        \item Let us consider when $a_i$ is such that $\gamma_{Ag} = \texttt{straight}$. 
        \begin{enumerate}
            \item This follows from the same arguments presented in \ref{safety_incomparable_straight}.
        \end{enumerate}
        \item Let us consider when $a_i$ is such that 
        $\gamma_{Ag} \in$ \\ $\{\texttt{left-lane change$,
        $right-lane change}\}$. 
        \begin{enumerate}
            \item $\text{Ag}$ will never select such an action at an intersection since $O_{\text{intersection lane-change}}(s,a,u)$ will evaluate to $\texttt{F}$.
        \end{enumerate}
        \item Let us consider when $a_i$ is such that $\gamma_{Ag} \in \\ \{\texttt{left-turn, right-turn}\}$.
        \begin{enumerate}
            \item By the assumption that all other agents are following the Agent Protocol, all $\text{Ag}'$ that are in bundle perpendicular to $Bu(Ag)$ will not be in the intersection and will not collide with $\text{Ag}$.
            \label{safety_incomparable_no_collision_turn}
            \item Further, the traffic light oracle $O_{\text{traffic light}}(s,a,u) = \texttt{T}$ only when $\lnot((\text{Ag}, a_i) \bot \text{Ag}')$ when $\gamma_{\text{Ag}} = \texttt{right-turn}$.
            \label{safety_incomparable_right_turn}
            \item Thus, when $\gamma_{\text{Ag}} = \texttt{right-turn}$ proof by \ref{safety_incomparable_no_collision_turn} and \ref{safety_incomparable_right_turn}.
            \item For an action $a_i$ where $\gamma_{\text{Ag}} = \texttt{left-turn}, \text{Ag}$ will only take $a_i$ if $O_{\text{traffic-light}}(s,a,u) = \texttt{T}$ and \\
            $O_{\text{unprotected left-turn}}(s,a,u) = \texttt{T}$.
            \item Since all agents are following the law based on Proof \ref{safety_proof}, $O_{\text{traffic light}}(s,a,u)=\texttt{T}$ means action will not cause the agent to collide with or violate the safety of the backup plan in perpendicular bundles.
            \item By the definition of the unprotected-left-turn oracle, $Ag$ will only take the left-turn action when it does not violate the safety of the backup plan of agents in oncoming traffic.
    \end{enumerate}
\end{enumerate}
\end{enumerate}
\end{proof}

\subsection{Safety Proof}
\begin{theorem}
Given all agents $\text{Ag} \in \mathfrak{A}$ in the quasi-simultaneous game select actions in accordance to the Agent Protocol specified in Section \ref{section_agent_protocol}, we can show the safety property $P \Rightarrow \square Q$, where the assertion $P$ is an assertion that the state of the game is such that $\forall Ag, S_{\text{Ag}, bp}(s,u) = \texttt{T}$, i.e. each agent has a backup plan action that is safe, as defined in \ref{definition_safety_back_up_plan_action}. We denote $P_t$ as the assertion over the state of the game at the beginning of the time-step $t$, before agents take their respective actions. $Q$ is the assertion that the agents never occupy the same grid point in the same time-step (e.g. collision never occurs when agents take their respective actions during that time-step). We denote $Q_t$ as the assertion for the agent states/actions taken at time-step $t$.
\end{theorem}

\begin{proof}
\label{safety_proof}
To prove an assertion of this form, we need to find an invariant assertion $I$ for which i) $P \Rightarrow I$, ii) $I \Rightarrow \square I$, and iii) $I \Rightarrow Q$ hold.
 We define $I$ to be the assertion that holds on the actions that agents select to take at a time-step. We denote $I_t$ to be the assertion on the actions agents take at time $t$ such that $\forall Ag$, $\text{Ag}$ takes $a\in Act_{\text{Ag}}$ where 1) it does not collide with other agents and 2) $\forall \text{Ag}, S_{\text{Ag}, bp}(u') = \texttt{T}$ where $s' = \tau_{\text{Ag}}(s, a)$, and $u'$ is the corresponding global state of the game after $\text{Ag}$ has taken its action $a$.
 
It suffices to assume:
\begin{enumerate}
    \label{safety_assumption:one_app}
    \item Each $\text{Ag} \in \mathfrak{A}$ has access to the traffic light states.
    \label{safety_assumption:two_app}
    \item There is no communication error in the conflict requests, token count queries, and the agent intention signals.
    \label{safety_assumption:three_app}
    \item All intersections in the road network $R$ are governed by traffic lights. 
    \item The traffic lights are designed to coordinate traffic such that if agents respect the traffic light rules, they will not collide.
    \label{safety_assumption:four_app}
    \item Agents follow the agent dynamics defined in Section \ref{agent_attributes}.
    \label{safety_assumption:five_app}
    \item For $t=0$, $\forall \text{Ag} \in \mathfrak{A}$ in the quasi-simultaneous game is initialized to: 
        \begin{itemize}
            \item Be located on a distinct grid point on the road network.
            \item Have a safe backup plan action $a_{bp}$ such that $S_{\text{Ag}, bp}(s, u) = \texttt{T}$. 
        \end{itemize}
        \label{safety_assumption:six_app}
        
\end{enumerate}

We can prove $P \Rightarrow \square Q$ by showing the following:
\begin{enumerate}
    \item $P_t \Rightarrow I_t$. This is equivalent to showing that if all agents are in a state where $P$ is satisfied at time $t$, then all agents will take actions at time $t$ where the $I$ holds. 
    \begin{enumerate}
        \item In the case that the assertion $P_t$ holds, let us show that $\text{Ag}$ will only choose an action $a \in Act_{Ag}$ for which the following two conditions hold: 1) $\mathcal{G}_{Ag}(s, a) \cap (\cup_{Ag'\in S}\mathcal{G}_{\text{Ag}'}(s', a')) = \emptyset$ and 2) $\forall Ag'\in S$, $\lnot ((\text{Ag}, a)\bot \text{Ag}')$, where the set $S$ is:
            \begin{enumerate}
            \label{safety_outside_bubble_app}
            \item The set $S \triangleq \{\text{Ag}'|\text{Ag} \prec \text{Ag}'\}$, i.e. agents with higher precedence than $\text{Ag}$. Proof by Lemma \ref{lemma_safety_higher}.
            \label{safety_higher_precedence_app}
            \item $S \triangleq \{\text{Ag}' | \text{Ag}' \prec \text{Ag} \}$, i.e. agents with lower precedence than $\text{Ag}$. Proof by Lemma \ref{lemma_safety_lower}.
            \label{safety_lower_precedence_app}
            \item $S \triangleq \{\text{Ag}' | \text{Ag}' \sim \text{Ag} \}$, i.e. agents with equal precedence than the agent. Proof by Lemma \ref{lemma_safety_equal}.
            \label{safety_equal_precedence_app}
            \item $S \triangleq \{\text{Ag}' | \text{Ag}' \not \sim \text{Ag} \}$, i.e. agents with precedence incomparable to the agent. Proof by Lemma \ref{lemma_safety_incomparable}.
            \label{safety_incomparable_precedence}
        \end{enumerate}
        \item The set of all agents, agents with lower precedence, higher precedence, equal precedence, and incomparable precedence, is complete and includes all agents.
        \label{safety:agent_completeness_app}
        \item By \ref{safety_top_two_tiers}-\ref{safety_incomparable_precedence} and \ref{safety:agent_completeness_app}, an agent will not take an action that will cause collision with any other agents (including itself) or violate the safety of the safety backup plan of all other agents, and thus any action taken by any agent will be such that following the action, the assertion $P$ still holds.  
    \end{enumerate}
    \item $P_t \Rightarrow I_t$. This is equivalent to showing that if all agents are in a state where $P$ is satisfied at time $t$, then all agents will take actions at time $t$ where the $I$ holds. This can be proven using arguments based on the design of the Agent Protocol. More details can be found in Lemmas A.\ref{lemma_safety_agents_outside_bubble}-A.\ref{lemma_safety_incomparable} in the Appendix.
    \label{p_implies_i_app}
    \item $I \Rightarrow \square I$. If agents take actions at time $t$ such that the assertion $I_t$ holds, then by the definition of the assertion $I$, agents will end up in a state where at time t+1, assertion $P$ holds, meaning $I_t \Rightarrow P_{t+1}$. Since $P_{t+1} \Rightarrow I_{t+1}$, from \ref{p_implies_i_app}, we get $I \Rightarrow \square I$.
    \item  $I \Rightarrow Q$. This is equivalent to showing that if all agents take actions according to the assertions in $I$, then collisions will not occur. This follows from the invariant assertion that agents are taking actions that do not cause collision, and the fact that all $\text{Ag}$ have a safe backup plan action $a_{bp}$ to choose from, and thus will always be able to (and will) take an action from which it can avoid collision in future time steps.  
\end{enumerate}
\end{proof}

\subsection{Liveness Lemmas}
\begin{lemma}
If the only $a \in Act_{Ag}$ for an agent $\text{Ag}$ for which $O_{\text{destination reachability}}(s,a,u) = \texttt{T}$ and $O_{\text{forward progress}}(s,a,u) = \texttt{T}$ is an action such that: $\gamma_{\text{Ag}} \in \{ \texttt{right-turn, left-turn}\}$ and the grid-point $s_f = \tau_{\text{Ag}}(s, a)$ is unoccupied (for a left-turn, where $a$ is the final action of the left-turn maneuver), $\text{Ag}$ will always eventually take $a$.
\label{liveness_lemma_left_right_turns}
\end{lemma}
\begin{proof} 
W.l.o.g., let us consider agent $\text{Ag} \in \mathfrak{A}$ in the quasi-simultaneous game $\mathfrak{G}$. We prove this by showing that all criteria required by the Agent Protocol are always eventually satisfied, thereby allowing $\text{Ag}$ to take action $a$.
\begin{enumerate}
    \item By the definition of $\mathfrak{R}$ and the agent dynamics, when $\text{Ag}$ is in a position where only $\gamma_{\text{Ag}} \in \{\texttt{right-turn}, \texttt{left-turn}\}$, it will neither send nor receive requests from other agents and $\mathcal{F}_{\text{Ag}}(u, a_i)$ will never be set to \texttt{T}.
    \item In accordance with the Action Selection Strategy, for $\text{Ag}$ to take action $a$, all the oracles in the Agent Profile must be simultaneously satisfied (so it will be selected over any other $a'\in Act_{\text{Ag}}$). Thus, we show: 
    \begin{enumerate}
        \label{liveness_turns_all_oracles}
        \item The following oracle evaluations will always hold when $\text{Ag}$ is in this state: $O_{\text{traffic intersection lane-change}}(s,a,u)=\texttt{T}$,$O_{\text{legal orientation}}(s,a,u)=\texttt{T}$, $O_{\text{static safety}}(s,a,u) = \texttt{T}$ and \\ $ O_{\text{traffic intersection clearance}}(s,a,u) = \texttt{T}$.
        \begin{enumerate}
            \item The first oracle is true vacuously and the following are true by the road network constraints and agent dynamics, Assumption 8, and the assumption in the lemma statement that $s_f = \tau(s,a)$ is unoccupied respectively.
        \end{enumerate}
        \item To show that the following oracles will always eventually simultaneously hold true, let us first consider when $\gamma = \{\texttt{right-turn}\}$.
        \begin{enumerate}
            \item By the assumption, the traffic light is red for a finite time, and when the traffic light is green, $O_{\text{traffic light}}(s,a,u) = \texttt{T}$. 
            \item $O_{\text{unprotected left-turn}}(s,a,u)$ is vacuously true for a right-turn action.
            \item Since $O_{\text{traffic intersection clearance}}(s,a,u) = \texttt{T}$ and by the safety proof \ref{safety_proof}, all $\text{Ag}$ are only taking actions in accordance with traffic laws so there will never be any $\text{Ag}'\in\mathfrak{A}$ blocking the intersection, making $O_{\text{dynamic safety}}(s,a,u) = \texttt{T}$.
            \item Thus, all oracles are always eventually simultaneously satisfied and $\text{Ag}$ can take $a$ where $\gamma = \{\texttt{right-turn}\}$
        \end{enumerate}
        \item Let us consider when $\gamma_{\text{Ag}} = \{\texttt{left-turn} \}$.
        \label{liveness_turn_left}
        \begin{enumerate}
            \item By Assumption 7, traffic lights are green for a finite time.
            \item By the safety proof \ref{safety_proof}, all $\text{Ag}$ are only taking actions in accordance with traffic laws so there will never be any $\text{Ag}'\in\mathfrak{A}$ blocking the intersection. 
            \item When $\gamma_{\text{Ag}} = \texttt{left-turn}$, by definition of the unprotected left-turn oracle, $\square \lozenge O_{\text{unprotected left-turn}}(s,a,u)$, specifically when the traffic light switches from green to red and $\text{Ag}$ has been waiting at the traffic light. 
            \item Thus, $\square \lozenge O_{\text{unprotected left-turn}}(s,a,u)$ after the light turns from green to red. 
            \item Further, $O_{\text{unprotected left-turn}}(s,a,u) = \texttt{T}$ combined with \\ $O_{\text{traffic intersection clearance}}(s,a,u) = \texttt{T}$ implies \\$O_{\text{dynamic safety}}(s,a,u) = \texttt{T}$.
            \item Thus, all oracles are always eventually simultaneously satisfied and $\text{Ag}$ can take $a$ where $\gamma = \{\texttt{left-turn}\}$.
        \end{enumerate}
    \end{enumerate}
    \item Thus, we have shown all oracles in the Agent Profile will always eventually be satisfied, and $\text{Ag}$ will take $a$ such that $O_{\text{destination reachability}}(s,a,u)= \texttt{T}$ and \\ $O_{\text{forward progress}}(s,a,u)= \texttt{T}$. 
\end{enumerate}
\end{proof}

\begin{lemma}
If the only $a \in Act_{\text{Ag}}$ for which \\ $O_{\text{destination reachability}}(s,a,u)= \texttt{T}$ and $O_{\text{forward progress}}(s,a,u)= \texttt{T}$ is when $a$ has \\ 
$\gamma_{\text{Ag}} \in \{ \texttt{right-lane change, left-lane change}\}$ and the grid-point(s) $\mathcal{G}(s,a)$ is (are) either unoccupied or agents that occupy these grid points will always eventually clear these grid points, $\text{Ag}$ will always eventually take this action $a$.
\label{lane_change_liveness_lemma}
\end{lemma}
\begin{proof}
W.l.o.g., let us consider agent $\text{Ag} \in \mathfrak{A}$ in the quasi-simultaneous game $\mathfrak{G}$. We prove this by showing that all criteria required by the Agent Protocol are always eventually satisfied, thereby allowing $\text{Ag}$ to take its action $a$.
\begin{enumerate}
    \item Let us consider Case A, when $a$ is such that $s_f = \tau_{\text{Ag}}(s,a) = \texttt{Goal}_{\text{Ag}}$, i.e. the action takes the agent to its goal, and let us show that $\text{Ag}$ will always eventually be able to take $a$.
    \label{liveness_lane_change_case_a}
    \item In accordance with the Action Selection Strategy, for $\text{Ag}$ to take $a$ is that 1) all the oracles in the agent profile must be simultaneously satisfied (so the action $a$ is chosen over any other $a'\in Act_{\text{Ag}}$, 2) $\mathcal{F}_{\text{Ag}}(u, a_i)=0$, and 3) $W_{\text{Ag}} = \texttt{T}$. 
    \item We first show all the oracles for $\text{Ag}$ will always be simultaneously satisfied: 
    \label{liveness_lane_change_case_a_oracles}
    \begin{enumerate}
        \item When $\text{Ag}$ is in this state, the following oracle evaluations always hold: $O_{\text{traffic light}}(s,a,u) = \texttt{T}$, \\ $O_{\text{traffic intersection lane-change}}(s,a,u)=\texttt{T}$, \\ $O_{\text{unprotected left turn}}(s,a,u)=\texttt{T}$, \\ $\square \lozenge O_{\text{traffic intersection clearance}}(s,a,u)$, $O_{\text{static safety}}(s,a,u) = \texttt{T}$, $O_{\text{traffic orientation}}(s,a,u)=\texttt{T}$.
        \begin{enumerate}
            \item The first four hold vacuously, the others hold by Assumption 8, and the last holds by Agent dynamics and the Road Network.
        \end{enumerate}
        \item $O_{\text{dynamic safety}}(s,a,u)=\texttt{T}$.
        \begin{enumerate}
            \item By the definition Road Network $\mathfrak{R}$, agent dynamics in Section \ref{agent_attributes}, and the condition that $\forall \text{Ag} \in \mathfrak{A}$ will leave $\mathfrak{R}$ (i.e. $\text{Ag}$ does not occupy any grid point on $\mathfrak{R}$ when it reaches its respective goal $\texttt{Goal}_{\text{Ag}}$). Thus, \\ $O_{\text{dynamic safety}}(s,a,u)=\texttt{T}$ whenever an agent is in this state. 
        \end{enumerate}
    \end{enumerate}
    \item In accordance with the action selection strategy, for $\text{Ag}$ to take $a$, it must be that $\mathcal{F}_{\text{Ag}}(u, a_i)=0$, i.e. the max-yielding-flag-not-enough must not be set. Let us show that this is always true.
    \label{liveness_lane_change_case_a_flag}
    \begin{enumerate}
        \item The only $\text{Ag}'$ that can cause the $\mathcal{F}_{\text{Ag}}(u, a_i)=1$ of $\text{Ag}$ is when an agent $\text{Ag}'$ is in a state where $La(Ag') = \texttt{Goal}_{\text{Ag}}$.
        \item W.l.o.g. let us consider such an $\text{Ag}'$. By liveness Assumption 9, upon approaching the goal, the agent $\text{Ag}'$ must be in a state where $\text{Ag}'$ backup plan action $a_{bp}$ will allow it to a complete stop before reaching its goal. 
        \label{complete_stop_goal}
        \item By \ref{complete_stop_goal}, $\text{Ag}'$ will always be in a state for which the max-yielding-not-enough flag for $\text{Ag}$ is $\mathcal{F}_{\text{Ag}}(u, a_i)=0$.
    \end{enumerate}
    \item In order for $\text{Ag}$ to take $a$, it must be that $W_{Ag}=1$. Let us show that this is always eventually true.
    \label{liveness_lane_change_case_a_winning}
    \begin{enumerate}
        \item In the case that $\text{Ag}$ has the maximum number of tokens, $\mathcal{W}_{Ag} = 1$ and $\text{Ag}$ will be able to take its forward action since all criteria are satisfied.
        \item Any $\text{Ag}' \in \mathcal{C}_{Ag}$ will be of equal or lower precedence than $\text{Ag}$.
        \item Any $\text{Ag}'$  with the maximum number of tokens will move to its goal since $\mathcal{W}_{\text{Ag}}=1$ and all the other criteria required for that agent to take its action will be true.
        \item By definition of the Action Selection Strategy in Section \ref{section_action_selection_strategy}, any agent $\hat{\text{Ag}}$ that replaces $\text{Ag}'$ will have taken a forward progress action and its respective token count will reset to 0.
        \label{agent_reset_to_0}
        \item Thus, any $\text{Ag}'$ will be allowed to take its action before $\text{Ag}$, but $\text{Ag}$'s token count $\texttt{Tc}_{\text{Ag}}$ will increase by one for every time-step this occurs. 
        \label{highest_token_count}
        \item Thus, by \ref{agent_reset_to_0} and by \ref{highest_token_count}, $\text{Ag}$ will always eventually have the highest token count in its conflict cluster such that $W_{\text{Ag}}=1$.
    \item Since conditions \ref{liveness_lane_change_case_a_oracles} and \ref{liveness_lane_change_case_a_flag} are always true, and \ref{liveness_lane_change_case_a_winning} is always eventually true, then all conditions will simultaneously always eventually be true and the $\text{Ag}$ will always eventually take the action $a$. 
    \end{enumerate}
    
    \item Let us consider Case B, when $a$ is the final action to take for an agent to reach its sub-goal (i.e. a critical left-turn or right-turn tile), and let us show $\text{Ag}$ will always eventually be able to take a forward progress action where $\gamma_{\text{Ag}} \in \{ \texttt{left-lane change}, \texttt{right-lane change}\}$. 
    \item In accordance with the Action Selection Strategy, for $\text{Ag}$ to take $a$ is that 1) $W_{\text{Ag}}=1$, 2) $\mathcal{F}_{\text{Ag}}(u, a_i)=0$, i.e. the max-yielding-flag-not-enough must not be set and 3) all the oracles in the Agent Profile must be simultaneously satisfied.
    \item Let us first consider when $W_{\text{Ag}}=1$, then $\square W_{\text{Ag}}$ until $\text{Ag}$ takes its forward progress action $a$ because by definition of $W_{\text{Ag}}$, $\text{Ag}$ has the highest token count in its conflict cluster, $\text{Ag}.\texttt{tc} = \text{Ag}.\texttt{tc}+1$, while $\text{Ag}$ does not select $a$ (and thus does not make forward progress) and any $\text{Ag}$ that newly enters $\text{Ag}$'s conflict cluster will have a token count of 0. 
    \label{liveness_lane_change_winning_true_beg}
    \item All the oracles are either vacuously or trivially satisfied by the assumptions except for $O_{\text{dynamic safety}}(s,a,u)$.
    \item By the Assumption 7, the traffic light will always cycle through red-to-green and green-to-red at the intersection $\text{Ag}$ is located at. 
    \item By the Assumption on the minimum duration of the red traffic light, all $\text{Ag}'$ will be in a state such that $\mathcal{F}_{\text{Ag}}(u, a_i) = 0$.
    \item By the lemma assumption that all $Ag'$ occupying grid points will always eventually take their respective forward progress actions, $\square \lozenge O_{\text{dynamic safety}}(s,a,u)$.
    \item Thus, all criteria for which $\text{Ag}$ can take its forward progress action $a$ will be simultaneously satisfied. 
    \item When $W_{\text{Ag}}$ = 0, we must show $\square \lozenge W_{\text{Ag}}$.
    \label{liveness_lane_change_winning_true_end}
    \begin{enumerate}
        \item For $\text{Ag}$, all agents in its conflict cluster have equal or lower precedence and are not in the same lane as $\text{Ag}$. 
        \item For any such $\text{Ag}'$ with equal precedence, $\text{Ag}'$ will always eventually take its forward progress action by the arguments in \ref{liveness_lane_change_winning_true_beg}-\ref{liveness_lane_change_winning_true_end} if $\text{Ag}'$ intends to make a lane-change.
        \item By the lemma assumption, any agents $Ag'$ occupying the grid points that $Ag$ needs to take its action will always eventually take its forward progress action so $\square \lozenge O_{\text{dynamic safety}}(s,a,u)$.
        \item Any $\hat{\text{Ag}}$ with lower precedence and higher token count that $\text{Ag}$ will take $\text{Ag}'$'s position and in doing so will have a token count of 0 and any $\text{Ag}$ that replaces any agents with higher token count than $\text{Ag}$ and is in $\text{Ag}$'s conflict cluster will have token count 0.
        \item Thus $\square \lozenge W_{\text{Ag}}$.
    \end{enumerate}

\end{enumerate}
\end{proof}

\begin{lemma}
Let us consider a road segment $rs \in RS$ where there exist grid points $g \in \mathcal{S}_{\text{sinks}}$. Every $\text{Ag} \in rs$ will always eventually be able to take $a \in Act_{\text{Ag}}$ for which $O_{\text{forward progress}}(s,a,u) = \texttt{T}$.
\label{lemma_destination_node}
\end{lemma}
\begin{proof}
We prove this by induction. W.l.o.g, let us consider $\text{Ag} \in \mathfrak{A}$. Let $m_{\text{Ag}}=\text{proj}_{\text{long}}(\texttt{Goal}_{\text{Ag}}) -\text{proj}_{\text{long}}(Ag.s)$.
\begin{enumerate}
    \item Base Case: $m_{\text{Ag}}=1$, i.e. $\text{Ag}$ only requires a single action $a$ to reach its goal $\texttt{Goal}_{\text{Ag}}$.
    \begin{enumerate}
        \item If $a$ is such that \\
        $\gamma_{Ag} \in \{ \texttt{left-lane change, right-lane change} \}$, then $\text{Ag}$ will take always eventually this action by Lemma \ref{lane_change_liveness_lemma}.
        \item If $a$ is such that $\gamma_{\text{Ag}} = \texttt{straight}$:
        \item In accordance with the Action Selection Strategy, for $\text{Ag}$ to take $a$ is that 1) all the oracles in the agent profile must be simultaneously satisfied (so the action $a$ is chosen over any other $a'\in Act_{\text{Ag}}$, and 2) $W_{\text{Ag}}$ = 1. 
        \item First, we show that all oracles in the agent profile will always be simultaneously satisfied. 
        \label{liveness_base_case_road_seg_0}
        \begin{enumerate}
            \item These all follow from the same arguments presented when $\gamma_{Ag} = \{\text{right-lane change}, \text{left-lane change}\}$ in Case A in Lemma \ref{lane_change_liveness_lemma}.
        \end{enumerate}
        \item In accordance with the Action Selection Strategy, we must show that $\square \lozenge W_{\text{Ag}}$. This is vacuously true since no $\text{Ag}$ will be in the agent's conflict cluster when an agent is in this state. 
        \label{liveness_base_case_road_seg_0_winning}
    \end{enumerate}
    \item Case $m=N$: Let us assume that any $\forall \text{Ag}$ where $m_{\text{Ag}} = N$ always eventually take $a \in Act_{\text{Ag}}$ for which \\
    $O_{\text{forward progress}}(s,a,u)=\texttt{T}$. 
    \label{liveness_road_seg_0_assumption}
    \item Case $m=N+1$: Let us show $\forall \text{Ag}$ where $m_{\text{Ag}}=N+1$ always eventually take $a$ for which \\
    $O_{\text{forward progress}} = \texttt{T}$.
    \label{liveness_road_seg_0_case_n_1}
    \begin{enumerate}
        \item Any $\text{Ag}$ for which $m_{\text{Ag}} > 1$ will always have an $a$ where $\gamma_{\text{Ag}} = \texttt{straight}$ such that $O_{\text{forward progress}}(s,a,u) = \texttt{T}$. 
        \item Thus, we show that $\text{Ag}$ always eventually will take $\gamma_{\text{Ag}} = \texttt{straight}$ such that $O_{\text{forward progress}}(s,a,u) = \texttt{T}$.
        \item W.l.o.g., let us consider $\text{Ag}$ for which $m_{\text{Ag}} = N+1$.
        \item In accordance with the Action Selection Strategy, for $\text{Ag}$ to take $a$ is 1) $W_{\text{Ag}}$ = 1 and 2) all the oracles in the agent profile must be simultaneously satisfied (so the action $a$ is chosen over any other $a'\in Act_{\text{Ag}}$).
        \item In accordance with the Action Selection Strategy, we must show $\square \lozenge W_{\text{Ag}}$.
        \label{always_eventually}
        \begin{enumerate}
            \item Any $\text{Ag}' \in \mathcal{C}_{\text{Ag}}$ will be an agent of equal or higher precedence and in separate lane. 
            \item Any such agent with higher token count than $\text{Ag}$ that is in its conflict cluster will always eventually be able to go by the inductive assumption in \ref{liveness_road_seg_0_assumption}.
            \item After all such agents take a forward progress action, they will no longer be in $\text{Ag}$'s conflict cluster and $\text{Ag}$ will have the highest token count since all $Ag$ that newly enter the conflict cluster will have token count of $0$. 
        \end{enumerate}
        \item After the assignment $W_{\text{Ag}} = 1$,  $\square W_{\text{Ag}}$ until $\text{Ag}$ selects $a$. This is true because by definition of $W_{\text{Ag}}$, $\text{Ag}$ has the highest token count in its conflict cluster, $\text{Ag}.\texttt{tc} = Ag.\texttt{tc}+1$, while $\text{Ag}$ does not select $a$, and any $\text{Ag}$ that enters $\text{Ag}$'s conflict cluster will have a token count of 0. 
        \item Let us show that the oracles in the Agent Profile will always evaluate to \texttt{T}.
        \begin{enumerate}
            \item The same arguments hold here as in Lemma \ref{lane_change_liveness_lemma}.\ref{liveness_lane_change_case_a} for all oracles except for \\ $O_{\text{dynamic safety}}(s,a,u)$, where  $\square \lozenge O_{\text{dynamic safety}}(s,a,u) = \texttt{T}$ by the inductive Assumption \ref{liveness_road_seg_0_assumption}.
        \end{enumerate}
    \end{enumerate}
\end{enumerate}
\end{proof}

\begin{lemma}
Let $\text{Ag}$ be on a road segment $rs \in RS$, where $RS$ is the set of nodes in the dependency road network dependency graph $\mathcal{G}_{\text{dep}}$. Let $rs$ be a road segment for which $\forall rs' \in RS s.t. \exists e: (rs', rs)$. Each road segment $rs'$ has vacancies in the grid points where $\text{Ag} \in rs$ would occupy if it crossed the intersection (i.e. $s_f = \tau_{\text{Ag}}(s,a)$), and we show that $\text{Ag}$ will always eventually take an action $a \in Act_{\text{Ag}}$ where $O_{\text{progress oracle}}(s,a,u) = \texttt{T}$.
\label{lemma_any_road_segment}
\end{lemma}
\begin{proof}
We prove this with induction. W.l.o.g., let us consider $\text{Ag} \in \mathfrak{A}$. Let $m_{\text{Ag}}=\text{proj}_{\text{long}}(g_{\text{front of rs}}) -\text{proj}_{\text{long}}(\text{Ag}.s)$, where $g_{\text{front of intersection}}$ represents a grid point at the front of the road segment.
\begin{enumerate}
    \item Base Case $m_{\text{Ag}} = 0$: Let us consider an $\text{Ag}$ whose next action will take will bring $\text{Ag}$ to cross into the intersection and show that $\text{Ag}$ will always eventually take $a$ for which $O_{\text{forward progress}}(s,a,u) = \texttt{T}$.
    \label{base_case_other_road_segment}
    \begin{enumerate}
        \item If the only $a$ where $O_{\text{forward progress}} = \texttt{T}$ is such that $\gamma_{\text{Ag}} \in \{\texttt{left-turn}, \texttt{ right-turn}$\}, proof by Lemma \ref{liveness_lemma_left_right_turns}.
        \item If the only $a$ where $O_{\text{forward progress}}(s,a,u) = \texttt{T}$ is such that $\gamma_{\text{Ag}} = \texttt{straight}$. 
        \label{straight_action_liveness}
        \begin{enumerate}
            \item In accordance with the Action Selection Strategy, for $\text{Ag}$ to take $a$ is that 1) all the oracles in the Agent Profile must be simultaneously satisfied (so the action $a$ is chosen over any other $a'\in Act_{\text{Ag}}$, 2) $W_{\text{Ag}}$ = 1. 
            \begin{enumerate}
                \item $O_{\text{unprotected left-turn}}(s,a,u)= \texttt{T}$, \\
                $O_{\text{traffic intersection lane-change}}(s,a,u)= \texttt{T}$, \\ $O_{\text{static safety}}(s,a,u)= \texttt{T}$, \\
                $O_{\text{traffic intersection clearance}}(s,a,u) = \texttt{T}$ \\
                $O_{\text{legal orientation}}(s,a,u)= \texttt{T}$.
                \item The first two oracles are true vacuously, followed by Assumption 8, and by agent dynamics and the road network $\mathfrak{R}$ definition, respectively, and by the assumption in the lemma statement. 
                \item $\square \lozenge O_{\text{traffic light}}(s,a,u)$ by Assumption 7.
                \item $O_{\text{dynamic obstacle}}(s,a,u) = \texttt{T}$ because by the safety proof, all $\text{Ag}$ take $a\in Act_{\text{Ag}}$ that satisfy the first top tiers of the agent profile so there will be no $\text{Ag}'\in \mathfrak{A}$ that are in the intersection when the traffic light for $\text{Ag}$ is green. Thus, whenever $O_{\text{traffic light}}(s,a,u) = \texttt{T}$, then it $O_{\text{dynamic obstacle}}(s,a,u) = \texttt{T}$ as well.  
            \end{enumerate}
            \item $W_{Ag}=1$ vacuously since neither $\text{Ag}$ or any $\text{Ag}' \in \mathfrak{A}$ will send a conflict request at the front of the intersection since all $a_i$ must satisfy $O_{\text{traffic intersection lane-change}}(s,a,u)$ according to the Safety Proof in Section A\ref{safety_proof}.
        \end{enumerate}
        \item By the safety proof in \ref{safety_proof}, $\text{Ag}$ will only take $a\in Act_{\text{Ag}}$ that satisfy the top two tiers of the Agent Profile, so $\text{Ag}$ will not take an $a$ where\\
        $\gamma_{\text{Ag}} \in \{ \texttt{left-lane change}, \texttt{right-lane change}$\} into an intersection. 
    \end{enumerate}
    \item Case $m_{Ag} = N$: Let us assume that $\text{Ag}$ with $m_{\text{Ag}} = N$ will always eventually take $a \in Act_{\text{Ag}}$ for which \\
    $O_{\text{forward progress}}(s,a,u)=\texttt{T}$.
    \item Case $m_{Ag} = N+1$: Let us show that any $\text{Ag}$ that is at a longitudinal distance of $N+1$ from the destination will always eventually take $a$ for which $O_{\text{forward progress}}(s,a,u) = \texttt{T}$. 
    \begin{enumerate}
        \item Let us consider when $\text{Ag}$'s only $a$ such that
        \\$O_{\text{forward progress}}(s,a,u) = \texttt{T}$ is
        \\ $\gamma_{\text{Ag}} \in \{ \texttt{right-lane change}, \texttt{left-lane change}$ \}. 
        \item Although $Ag$ may not have priority (since it does not have max tokens in its conflict cluster), any $Ag$ that occupies grid points $\mathcal{G}(s,a,u)$ will always eventually make forward progress by Argument \ref{base_case_other_road_segment}.
        \item Further, 
        \item Once these agents have made forward progress, any $\hat{Ag}$ that replace $\text{Ag}'$ will have a $\texttt{Tc}_{Ag}=0$ and since $Ag$ is always increasing its token counts as it cannot make forward progress, it will always eventually have the max tokens and thus have priority over those grid points.
        \item Thus, this can be proven by using Case B in Lemma \ref{lane_change_liveness_lemma}.
        \item For all other $a\in Act_{\text{Ag}}$ are actions for which $\gamma_{\text{Ag}} = \texttt{straight}$, and the same arguments as in the proof of straight actions for $rs$ with $g \in \mathcal{S}_{\text{sinks}}$ in  \ref{liveness_road_seg_0_case_n_1} hold.  
        \label{straight_forward_progress_induction_n_plus_one}
    \end{enumerate}
\end{enumerate}
\end{proof}

\subsection{Liveness Proof}
\begin{theorem}[Liveness Under Sparse Traffic Conditions]
Under the Sparse Traffic Assumption given by \ref{assumptions_sparse_traffic} and given all agents $\text{Ag} \in \mathfrak{A}$ in the quasi-simultaneous game select actions in accordance with the agent protocol specified in Section \ref{section_agent_protocol}, liveness is guaranteed, i.e. all $\text{Ag} \in \mathfrak{A}$ will always eventually reach their respective goals.
\end{theorem}
\begin{proof}
It suffices to assume:
\begin{enumerate}
    \item $\forall \text{Ag} \in \mathfrak{A}$, $\forall \text{Ag}' \in \mathbb{B}_{\text{Ag}}$, $\text{Ag}$ knows $\text{Ag}'.s, Ag'.i$, i.e. the other agent's state $\text{Ag}.s$ and intended action $a_i$ and all $\text{Ag}$ within a region around the intersection defined in the Appendix. 
    \item Each $\text{Ag} \in \mathfrak{A}$ has access to the traffic light states. 
    \item There is no communication error in the conflict requests, token count queries, and the agent intention signals. 
    \item For $t=0$, $\forall \text{Ag} \in \mathfrak{A}$ in the quasi-simultaneous game is initialized to: 
        \begin{itemize}
            \item Be located on a distinct grid point on the road network.
            \item Have a safe backup plan action $a_{bp}$ such that $S_{\text{Ag}, bp}(u) = \texttt{T}$. 
        \end{itemize}
    \item The traffic lights are red for some time window $\Delta t_{\text{tl}}$ such that $t_{\text{min}}<\Delta t_{\text{tl}}<\infty$, where $t_{\text{min}}$ is defined in the Appendix in Section \ref{traffic_min_time}.   
    \item The static obstacles are not on any grid point $g$ where \\
    $g.d= 1$. 
    \item Each $\text{Ag}$ treats its respective goal $\text{Ag}.\texttt{g}$ as a static obstacle.
    \item Bundles in the road network $\mathfrak{R}$ have no more than 2 lanes.
    \item The road network $R$ is such that all intersections are governed by traffic lights. 
\end{enumerate}
and prove: 
\begin{enumerate}
    \item Let us consider a road segment $r \in RS$ that contains grid point(s) $g \in \mathcal{S}_{\text{sinks}}$. Every $\text{Ag} \in r$ will be able to always eventually take $a \in Act_{Ag}$ for which $O_{\text{forward progress}}(s,a,u) = \texttt{T}$.
    \item Let us consider a road segment $rs \in RS$. Let us assume $\forall rs \in RS, \exists (rs, rs') \in G_{\text{dep}}$, i.e. the clearance of $rs$ depends on the clearance of all $rs'$. We use inductive reasoning to show that any $\text{Ag}$ on $rs$ will always eventually take an $a \in Act_{\text{Ag}}$ where $O_{\text{forward progress}}(s,a,u) = \texttt{T}$.
    \item For any $\mathfrak{R}$ where the dependency graph $G_{\text{dep}}$ (as defined in \ref{road_network_dependency_graph}) is a directed-acylcic-graph (DAG), we prove all $\text{Ag} \in \mathfrak{A}$ will always eventually take $a\in Act_{Ag}$ for which \\ $O_{\text{forward progress}}(s,a,u)=\texttt{T}$ inductively as follows.
    \begin{enumerate}
        \item A topological sorting of a directed acyclic graph G = (V, E) is a linear ordering of vertices V such that $(u, v) \in E \rightarrow u$ appears before $v$ in ordering.
        \item If and only if a graph $G$ is a DAG, then $G$ has a topological sorting. Since $G_{\text{dep}}$ is a $DAG$, it has a topological sorting. 
        \item We can then use an argument by induction on the linear ordering provided by the topological sorting to show that all $\text{Ag}$ always eventually take $a\in Act_{Ag}$ for which \\ $O_{\text{forward progress}}(s,a,u)=\texttt{T}$.
        \begin{enumerate}
            \item Let $l$ denote the linear order associated with the road network dependency graph $G_{\text{dep}}$, where an ordering of $l=0$ denotes a road segment with source nodes.
            \item Base Case $l=0$. This can be proven true by Lemma \ref{lemma_destination_node}.
            \item Let us assume this is true for any road segment where $l=N$.
            \label{inductive_assumption_liveness_proof}
            \item Under the Inductive Assumption \ref{inductive_assumption_liveness_proof}, there will be clearance in any road segment that agent $\text{Ag}$ depends on for $\text{Ag}$ to make forward progress to its destination.
            \item Since all $\text{Ag}$ are following the traffic laws by the Safety proof in \ref{safety_proof}, the clearance spots will be given precedence to $\text{Ag} \in rs$ for a positive, finite time, and thus the assumptions required in Lemma \ref{liveness_lemma_left_right_turns} and \ref{lane_change_liveness_lemma} used to prove Lemma \ref{lemma_any_road_segment} will hold. 
            \item Thus, the Lemma \ref{lemma_any_road_segment} to show that all $\text{Ag}$ for which $l=N+1$ always eventually take an action for which \\
            $O_{\text{forward progress}}(s,a,u)=\texttt{T}$.
        \end{enumerate}
    \end{enumerate}
    \item When the graph $G_{\text{dep}}$ is cyclic, the Sparsity Assumption \ref{assumptions_sparse_traffic} can be used to prove all agents always eventually take an action for which $O_{\text{forward progress}}(s,a,u)=\texttt{T}$.
    \begin{enumerate}
        \item The sparsity assumption \ref{assumptions_sparse_traffic} ensures that there is at least one vacancy in any map loop. 
        \item Let us consider $Ag$ inside a map loop.
        \begin{enumerate}
            \item Let us consider $Ag$ in the loop for which the vacancy is directly ahead of $Ag$. If the vacancy is directly ahead of $Ag$, then if the only forward progress action $a$ keeps $Ag$ in the loop, $Ag$ will always eventually take its action by Lemmas \ref{liveness_lemma_left_right_turns}, \ref{lane_change_liveness_lemma} and the arguments in Lemma \ref{lemma_any_road_segment} \ref{straight_action_liveness}. If the only forward progress action $a$ makes $Ag$ leave the loop, $Ag$ will always eventually take its action by the sparsity assumption \ref{assumptions_sparse_traffic} and the inductive arguments in \ref{liveness_induction_statement}.
            \label{first_in_loop}
            \item By \ref{first_in_loop}, it can then be inductively shown that any $Ag$ in the loop will always eventually have a vacancy for which it can take a forward progress action. 
            \label{inductivity_for_agents_in_map_loop}
        \end{enumerate}
        \item Let us consider $Ag$ on a road segment that is not part of a map loop.        \begin{enumerate}
            \item Let us consider an action $a$ that takes $Ag$ into a map loop. If the grid point required by $Ag$ to make forward progress is occupied, by \ref{inductivity_for_agents_in_map_loop}, it will always eventually be unoccupied. If the only action $Ag$ can take is such that $\gamma_{Ag}=\{\texttt{lane-change}\}$ since all $Ag'$ in the loop are reset when they take forward progress action, $Ag$ will always eventually have the max token count. Thus, the same arguments in Lemma \ref{lane_change_liveness_lemma} hold. If the only action $Ag$ can take is such that $Ag$ crosses into an intersection, the traffic light rules ensure that $Ag$ has precedence over any $Ag$ in the loop. Thus, $Ag$ will always eventually take a forward progress action by Lemma \ref{liveness_lemma_left_right_turns} and  Lemma \ref{lemma_any_road_segment} \ref{straight_action_liveness}. 
            \item For any action $a$ that does not take $Ag$ into a map loop, $Ag$ can take a forward action because of the sparsity assumptions \ref{assumptions_sparse_traffic} and the inductive arguments in \ref{liveness_induction_statement}.
        \end{enumerate}
    \end{enumerate}
    \item By the induction arguments and by definition of the forward progress oracle $O_{\text{forward progress}}(s,a,u)$, all $\text{Ag}$ will always eventually take actions that allow them to make progress to their respective destinations, and liveness is guaranteed.
\end{enumerate}
\end{proof}

\subsection{Traffic Light Assumptions}
A traffic light grid point contains three states
$g.s=\{ \texttt{red}, \texttt{yellow}, \texttt{green} \}$. The traffic lights at each intersection are coordinated so that if all agents obey the traffic signals, collision will not occur (i.e. the lights for the same intersection will never be simultaneously green) and the lights are both red for long enough such that $\text{Ag}$ that entered the intersection when the light was $\texttt{yellow}$ will be able to make it across the intersection before the other traffic light turns $\texttt{green}$. 

\subsubsection{Traffic Light Minimum Time}
\label{traffic_min_time}
In order to guarantee that agents will always eventually be able to make a lane-change to a critical tile, the traffic light has to be red for sufficiently long such that any $Ag'$ that may cause $\mathcal{F}_{\text{Ag}}(u, a_i) = \texttt{T}$ is slowed down for long enough such that $Ag$ can take its lane-change action. This can be computed simply once given the dynamics of $\text{Ag}$. Normally a simple heuristic can be used instead of computing this specific lower-bound.

\subsection{Simulation Maps}
\begin{figure}[H]
\centering
\includegraphics[scale=0.15]{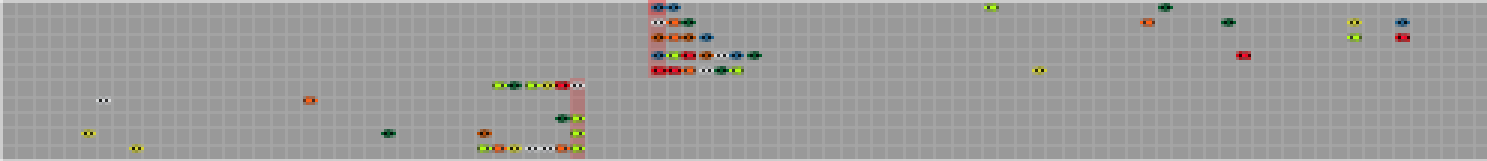}
\caption{Straight road map environment.}
\label{fig_straight_road}
\end{figure}

\begin{figure}[H]
\centering
\includegraphics[scale=0.5]{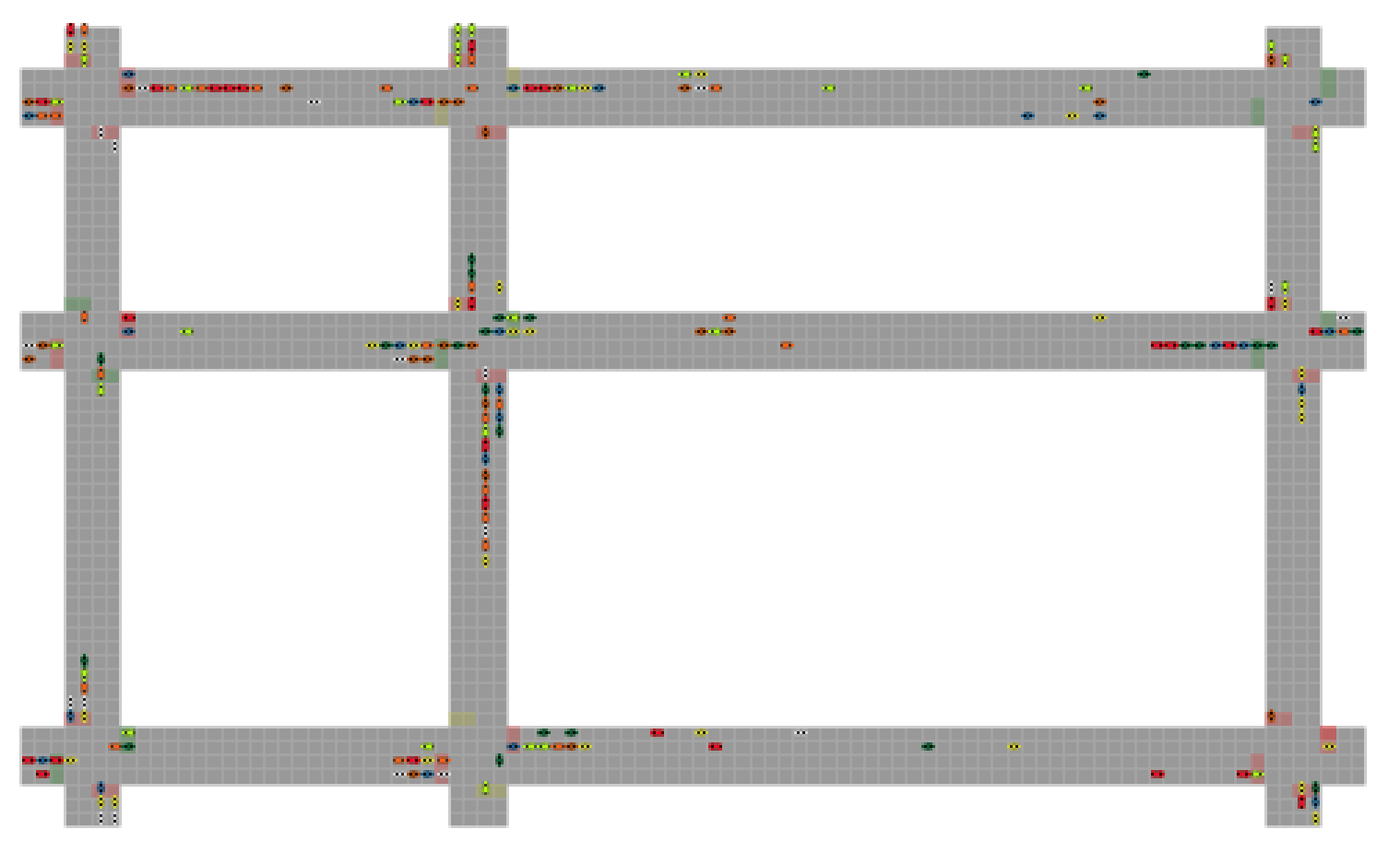}
\caption{City blocks map environment.}
\label{fig_city_blocks}
\end{figure}

\subsection{Simulation Environment Features}
A road network environment, complete with legal lane orientations, intersections, and traffic lights, can be specified via a CSV file. The specified (by the user) road network environment forms a map data structure graph, which decomposes the roads into bundles, mentioned in \ref{section_precedence_rules}.

The map will automatically parse the boundaries and lane directions of the road network to define where agents can either spawn from or exit the road network. In each game scenario, agents will randomly spawn according to a specified spawn rate.

Each agent has the following attributes in our simulation: parameters like min and max velocity and accelerations, dynamics specified by agent actions and their corresponding occupancy grids, goal location, agent color, ID, token count. Note, these attributes can be modified depending on what the user wants to include. For each agent, a graph-planning algorithm is used to compute a high-level motion plan on the map graph to get the agent to its goal.

Each game scenario is comprised of the road network graph and a set of agents (constantly changing over time as new agents spawn and old agents reach their goals and leave). The game is simulated forward for a specified number of time steps and the traces from the simulation are saved. The animation module in RoSE animates the traces from the simulated game. 

RoSE also offers a collection of debugging tools to help reconstruct scenarios that occurred during a simulated game. If the user would like to regenerate the same initialization, the simulation has a feature where users can specify a specific randomization seed. There is a configuration tool that allows users to prescribe the states of a set of agents and their respective goals. A final debugging tool outputs the variables of the agent that were relevant to the decision-making process. 
\end{appendices}

\end{document}